\documentclass[a4paper,11pt]{article}
\pdfoutput=1
\usepackage{a4wide}
\usepackage[UKenglish]{babel}
\usepackage[noadjust]{cite}
\usepackage{amsmath}
\usepackage{amssymb}
\usepackage{amsthm} 
\usepackage{bbm}
\usepackage{mathtools}
\usepackage{mathrsfs}
\usepackage{csquotes}
\usepackage{todonotes}
\usepackage[title]{appendix}
\usepackage{xcolor}
\usepackage{blindtext}
\usepackage{enumerate}
\usepackage{manfnt}
\usepackage{multicol}
\usepackage{microtype}

\MakeRobust{\ref}

\usepackage{stmaryrd}

\usepackage{dutchcal}

\usepackage[scr=boondox]{mathalfa}

\makeatletter
\newcommand{\reducedplus}{\mathpalette\reduced@plus\relax}
\newcommand{\reduced@plus}[2]{%
  \sbox6{$\m@th#1+$}%
  \sbox8{\scalebox{0.875}{\copy6}}%
  \dimen@=\dimexpr(\wd6-\wd8)/3\relax
  \raisebox{\dimen@}{\box8}%
}
\newcommand{\boxoperation}[2][\mathbin]{%
  #1{\mathpalette\box@operation{#2}}%
}
\newcommand{\box@operation}[2]{%
  \ooalign{$\m@th#1\boxempty$\cr\hidewidth$\m@th#1#2$\hidewidth\cr}%
}
\newcommand{\bplus}{\boxoperation{\reducedplus}}

\newcommand{\YES}{\textsc{Yes}}
\newcommand{\NO}{\textsc{No}}
\newcommand{\bigboxplus}{
  \mathop{
    \vphantom{\bigoplus} 
    \mathchoice
      {\vcenter{\hbox{\resizebox{\widthof{$\displaystyle\bigoplus$}}{!}{$\bplus$}}}}
      {\vcenter{\hbox{\resizebox{\widthof{$\bigoplus$}}{!}{$\bplus$}}}}
      {\vcenter{\hbox{\resizebox{\widthof{$\scriptstyle\oplus$}}{!}{$\bplus$}}}}
      {\vcenter{\hbox{\resizebox{\widthof{$\scriptscriptstyle\oplus$}}{!}{$\bplus$}}}}
  }\displaylimits 
}
\newcommand\ang[1]{{\ensuremath\langle #1\rangle}}
\newcommand\NP{{\operatorname{NP}}}

\DeclareSymbolFont{extraup}{U}{zavm}{m}{n}
\DeclareMathSymbol{\varheart}{\mathalpha}{extraup}{86}
\DeclareMathSymbol{\vardiamond}{\mathalpha}{extraup}{87}

\usepackage{hyperref} 
\hypersetup{colorlinks=true,citecolor=blue,linkcolor=blue,urlcolor=blue}

\newcommand{\Test}[2]{
\def\temp{#2}\ifx\temp\empty
  \operatorname{Test}_{#1}
\else
  \operatorname{Test}_{#1}^{#2}
\fi
}

\newcommand{\qTo}[1]{
  \to_{#1}
}

\newcommand{\qToMR}[1]{
  \xrightarrow{\operatorname{MR}}{}{\hspace{-.35em}}_{#1}
}

\newcommand{\C}{\mathbb{C}}

\newcommand{\K}{\mathbf{K}}

\newcommand{\X}{\mathbf{X}}
\newcommand{\Y}{\mathbf{Y}}

\newcommand{\N}{\mathbb{N}}
\newcommand{\R}{\mathbb{R}}

\newcommand{\HH}{\mathbb{H}}

\newcommand{\freeM}{\mathbb{F}_{\Mminion}}

\newcommand{\freeQH}{\mathbb{F}_{\Qminion_\HH}}

\renewcommand{\vec}[1]{\mathbf{#1}}

\newcommand{\bh}{\vec{h}}

\newcommand{\bu}{\vec{u}}
\newcommand{\bx}{\vec{x}}
\newcommand{\bv}{\vec{v}}
\newcommand{\by}{\vec{y}}
\newcommand{\bw}{\vec{w}}
\newcommand{\be}{\vec{e}}

\newcommand{\bq}{\vec{q}}
\newcommand{\bz}{\vec{z}}

\DeclareMathAlphabet{\mathbx}{U}{BOONDOX-ds}{m}{n}
\SetMathAlphabet{\mathbx}{bold}{U}{BOONDOX-ds}{b}{n}
\DeclareMathAlphabet{\mathbbx} {U}{BOONDOX-ds}{b}{n}

\DeclareMathOperator{\tr}{tr}

\DeclareMathOperator{\Proj}{Proj}
\DeclareMathOperator{\pr}{pr}

\DeclareMathOperator{\SDP}{SDP}
\DeclareMathOperator{\Pol}{Pol}
\DeclareMathOperator{\PCSP}{PCSP}
\DeclareMathOperator{\CSP}{CSP}

\DeclareMathOperator{\id}{id}
\DeclareMathOperator{\rg}{rg}

\DeclareMathOperator{\ar}{ar}

\DeclareMathOperator{\End}{End}

\newcommand{\Dminion}{\ensuremath{{\mathscr{D}}}}
\newcommand{\Sminion}{\ensuremath{{\mathscr{S}}}}
\newcommand{\Mminion}{\ensuremath{{\mathscr{M}}}}

\newcommand{\Nminion}{\ensuremath{{\mathscr{N}}}}
\newcommand{\Qminion}{\ensuremath{{\mathscr{Q}}}}

\newcommand{\Cminion}{\ensuremath{{\mathscr{C}}}}
  
\newcommand{\bzero}{\mathbf{0}}
\newcommand{\zeroOp}{\mathbx{0}}

\theoremstyle{plain}
\newtheorem{thm}{Theorem}
\newtheorem*{thm*}{Theorem}
\newtheorem{lem}[thm]{Lemma}
\newtheorem*{lem*}{Lemma}
\newtheorem{prop}[thm]{Proposition}
\newtheorem*{prop*}{Proposition}
\newtheorem{cor}[thm]{Corollary}
\newtheorem{conj}[thm]{Conjecture}
\newtheorem{question}[thm]{Question}
\theoremstyle{definition}
\newtheorem{defn}[thm]{Definition}
\newtheorem{rem}[thm]{Remark}

\newtheorem{example}[thm]{Example}

\begin{document}

\author{Lorenzo Ciardo\\
University of Oxford\\
\texttt{lorenzo.ciardo@cs.ox.ac.uk}
}

\title{Quantum advantage and CSP complexity\thanks{This work will appear in the Proceedings of the 39th Annual ACM/IEEE Symposium on Logic in Computer Science (LICS’24). The research leading to these results was supported by UKRI EP/X024431/1.}} 

\date{}

\maketitle

\begin{abstract}
\noindent Information-processing tasks modelled by homomorphisms between relational structures can witness quantum advantage when entanglement is used as a computational resource. We prove that the occurrence of quantum advantage is determined by the same type of algebraic structure (known as a minion) that captures the polymorphism identities of CSPs and, thus, CSP complexity. We investigate the connection between the minion of quantum advantage and other known minions controlling CSP tractability and width. In this way, we make use of complexity results from the algebraic theory of CSPs to characterise the occurrence of quantum advantage in the case of graphs, and to obtain new necessary and sufficient conditions in the case of arbitrary relational structures.
\end{abstract}

\section{Introduction}

The \emph{constraint satisfaction problem} (CSP) is a computational paradigm that provides a common framework for studying the complexity of a variety of combinatorial problems. The input consists in a set of variables, a set of possible values for the variables, and a set of constraints between the variables; the goal is to determine whether there exists an assignment of values to the variables that satisfies the constraints. In its full generality, the CSP is NP-complete. Nevertheless, restricting the admissible constraints to a selected class can make the problem solvable in polynomial time. This is the case of linear equations, which can be formulated as a CSP where all constraints are linear relations between the variables.
The theory of CSP complexity aims at explaining the complexity of a CSP over some specific type of constraints in terms of the structure of the constraints.

An elegant way of formalising the CSP is through the notion of homomorphisms between relational structures. A relational structure consists of a set of vertices and a set of relations on the vertices; for example, a graph is a relational structure having a single binary symmetric irreflexive relation. A homomorphism $\X\to\Y$ between two relational structures $\X$ and $\Y$ is a map between the vertex sets of $\X$ and $\Y$ that preserves all relations, in the sense that the image of a tuple in a relation of $\X$ must belong to the corresponding relation of $\Y$. We can then formulate the CSP \emph{parameterised} by $\Y$ as the problem of deciding if an arbitrary structure $\X$ admits a homomorphism to $\Y$. In this setting, the relations of $\X$ encode the tuples of variables that appear in some constraint of the CSP, while the relations of $\Y$ encode the types of admissible constraints. 

For example, linear equations over a finite field $\mathbb{F}$ of $n$ elements correspond to $\CSP(\Y)$, where the domain of $\Y$ is $\{1,\dots,n\}$, and the relations of $\Y$ are affine hyperplanes in vector spaces over $\mathbb{F}$. If $\Y$ is the $n$-clique $\K_n$, $\CSP(\Y)$ corresponds to another primary example of a CSP --- namely, the \emph{graph $n$-colouring} problem, which consists in testing whether the chromatic number of a graph $\X$ is at most $n$. 
A natural generalisation is the \emph{graph $\Y$-colouring} problem, which corresponds to $\CSP(\Y)$ for some fixed graph $\Y$. 
The complexity of graph $n$-colouring was characterised by Karp~\cite{Karp72}: If $n=2$, the problem is in $\operatorname{P}$; otherwise, it is $\NP$-complete. Hell and Ne\v{s}et\v{r}il~\cite{HellN90} extended the dichotomy by proving that graph $\Y$-colouring is in $\operatorname{P}$ if $\Y$ is bipartite, and it is $\NP$-complete otherwise.
In~\cite{Feder98:monotone}, Feder and Vardi proposed the conjecture that Karp's and Hell--Ne\v{s}et\v{r}il's dichotomies could in fact extend to all CSPs.
Feder--Vardi's Dichotomy Conjecture inspired a major research programme in the theory of CSPs, that culminated twenty years later with its positive resolution obtained by Bulatov~\cite{Bulatov17:focs} and, independently, by Zhuk~\cite{Zhuk20:jacm}.
\\

The advent of quantum computation led to the natural programme of understanding the type of advantage derived by having access to quantum resources --- as opposed to classical resources --- in computational tasks. In order to formalise this notion of \emph{quantum advantage} for computational tasks consisting in constraint satisfaction problems, it shall be convenient to adopt an alternative description of homomorphisms between relational structures, in terms of winning strategies for non-local games.

For the sake of simplicity, suppose that $\X$ and $\Y$ are two graphs. The \emph{$\X\Y$-homomorphism game}, introduced in~\cite{MancinskaRoberson16}, involves two players (Alice and Bob) and a Verifier. The Verifier sends Alice and Bob vertices $x_A$ and $x_B$ of $\X$, respectively. Alice and Bob (who cannot communicate during the game, but can agree on a strategy before the game starts) respond by sending the Verifier vertices $y_A$ and $y_B$ of $\Y$, respectively. They win the $\X\Y$-homomorphism game if $y_A$ and $y_B$ are adjacent (resp. equal) vertices of $\Y$ whenever $x_A$ and $x_B$ are adjacent (resp. equal) vertices of $\X$. If only resources of classical type are allowed, the existence of a perfect strategy for this game --- i.e., a strategy that makes Alice and Bob win with probability $1$ --- is equivalent to the existence of a homomorphism $\X\to\Y$. In fact, a classical perfect strategy simply consists in Alice and Bob sharing a protocol for selecting some homomorphism $f:\X\to\Y$, and responding with the vertex $f(x)$ whenever they are sent a vertex $x$ by the Verifier. 

Suppose now that Alice and Bob have access to quantum resources, modelled by an entangled state in the space $\HH_A\otimes\HH_B$, where $\HH_A\simeq\HH_B$ are finite-dimensional (real or complex) Hilbert spaces. This state is shared by Alice and Bob; the separation between the two players corresponds to the fact that Alice (resp. Bob) can only perform measurements on $\HH_A$ (resp. $\HH_B$). This extra resource allows the players to devise quantum-assisted strategies: Each player performs a positive operator-valued measure of their part of the state, which results in a vertex of $\Y$ that is then sent to the Verifier. A \emph{quantum perfect strategy} over the space $\HH=\HH_A\simeq\HH_B$ for the $\X\Y$-homomorphism game is a quantum-assisted strategy that makes Alice and Bob win with probability $1$. In this case, $\X$ and $\Y$ are said to be \emph{quantum homomorphic} over $\HH$. 
The extra correlation between the answers of Alice and Bob produced by the entanglement of the quantum state can result in strategies that are strictly stronger than classical strategies. In other words, it is possible for two graphs $\X$ and $\Y$ to be quantum homomorphic even when they are not homomorphic in the classical sense.
In~\cite{MancinskaRoberson16}, the question of understanding for which graphs $\Y$ the notions of quantum and classical homomorphisms differ was posed. More precisely, we say that $\Y$ \emph{has quantum advantage} if there exists some graph $\X$ such that $\X$ and $\Y$ are quantum homomorphic but not classically homomorphic.
\begin{question}[\cite{MancinskaRoberson16}]
\label{question_quantum_advantage_graphs}
    Characterise the class of graphs $\Y$ having quantum advantage.
\end{question}

The notions of quantum perfect strategy, quantum homomorphism, and quantum advantage that we consider in this work extend those described above, in that they apply to arbitrary relational structures as opposed to only graphs. This extension was recently introduced in~\cite{abramsky2017quantum}, and it involves a modified version of the $\X\Y$-homomorphism game discussed above, where the roles of Alice and Bob are asymmetric: Alice receives a tuple $\bx_A$ in some relation of $\X$, and returns a tuple $\by_A$ in the corresponding relation of $\Y$; Bob receives a vertex $x_B$ of $\X$ and returns a vertex $y_B$ of $\Y$. They win the game if, whenever $x_B$ appears in the tuple $\bx_A$, $y_B$ appears in the corresponding position of the tuple $\by_B$.\footnote{If $\X$ and $\Y$ are graphs, it was noted in~\cite{abramsky2017quantum} that this definition results in a stronger notion of quantum homomorphism than the one from~\cite{MancinskaRoberson16}. As we shall see in Section~\ref{sec_quantum_advantage_graphs}, the results of this paper apply to both definitions.}\\

The main conceptual contribution of this work is to establish that \emph{the occurrence of quantum advantage is governed by 
the same type of algebraic structure that determines the complexity of constraint satisfaction problems} --- namely, the algebraic structure known as a \emph{minion}. 
We start by discussing the role of minions in CSP complexity.

The mathematical framework that allowed attacking and, eventually, proving Feder-Vardi's Dichotomy
Conjecture was the so-called \emph{algebraic approach to CSPs}, whose origins can be traced back to~\cite{Jeavons97:closure,Jeavons98:algebraic,Pippenger02}. The gist of this approach is that the complexity of a CSP is entirely determined by a type of identities satisfied by its higher-order symmetries, known as \emph{polymorphisms}. A polymorphism of a relational structure $\Y$ is a homomorphism $\Y^\ell\to\Y$, where $\Y^\ell$ is the $\ell$-fold direct power of $\Y$. If $\Y$ and $\Y'$ are two structures such that the set $\Pol(\Y)$ of the  polymorphisms of $\Y$ is included in $\Pol(\Y')$ --- thus meaning that $\Y'$ has ``more symmetries'' than $\Y$ ---
there exists a specific kind of polynomial-time reduction (known as a \emph{gadget reduction}) from $\CSP(\Y')$ to $\CSP(\Y)$ (in symbols, $\CSP(\Y')\leq_{\operatorname{g}}\CSP(\Y)$).
This gives a powerful way of comparing the complexity of different CSPs, by looking at the properties of their polymorphisms. 
The implication
\begin{align*}
\Pol(\Y)\subseteq\Pol(\Y')\;\;\;\;\Rightarrow\;\;\;\;\CSP(\Y')\leq_{\operatorname{g}}\CSP(\Y)    
\end{align*}
can in fact be strengthened: As established in~\cite{BOP18}, in order for a gadget reduction to exist, it is enough for $\Pol(\Y')$ to satisfy all identities of a certain type (known as \emph{height-1} or \emph{minor} identities) that are satisfied in $\Pol(\Y)$. One can then endow $\Pol(\Y)$ with a particular algebraic structure (known as \emph{minion} or \emph{minor-closed class}~\cite{Pippenger02}) that encodes minor identities. This results in the stronger implication 
\begin{align*}
\Pol(\Y)\to\Pol(\Y')\;\;\;\;\Rightarrow\;\;\;\;\CSP(\Y')\leq_{\operatorname{g}}\CSP(\Y),    
\end{align*}
where ``$\to$'' denotes a minor-preserving map (or \emph{minion homomorphism}) between the minions $\Pol(\Y)$ and $\Pol(\Y')$. Furthermore, it was shown in~\cite{BOP18} that the converse of the implication above also holds.
In this sense, the polymorphism minion precisely captures CSP complexity.

\paragraph{Contributions}
Our first contribution is
to show that quantum advantage is a property of the polymorphism minion, just like CSP complexity.
To that end,
we construct a minion $\Qminion_\HH$ having the property that a structure $\Y$ has quantum advantage over $\HH$ if and only if $\Qminion_\HH\not\to\Pol(\Y)$. This \emph{quantum minion} belongs to a particular class of minions identified in~\cite{cz23soda:minions}, known as \emph{linear minions}. Its objects are tuples of orthogonal subspaces of the Hilbert space $\HH$, while the minors are defined in terms of sums of linear spaces.

Using this connection, the question ``\emph{Which structures have quantum advantage?}'' can be reformulated as ``\emph{Which minions admit a homomorphism from the quantum minion?}''. In other words, the problem of classifying the occurrence of quantum advantage is translated into the problem of locating $\Qminion_\HH$ within the preorder on the class of minions induced by the minion homomorphism relation. We prove the following three results:\\[-1em]

    \noindent$(i)$ If $\HH$ has dimension $1$ or $2$, $\Qminion_\HH$ is equivalent to the so-called \emph{dictator minion} $\Dminion$ (also known as projection minion). To show this, we make use of the Hopf fibration of the complex projective line $\C\textbf{P}^1$ into the $2$-sphere $S^2$, which has the property of mapping pairs of orthogonal vectors into antipodal points. Then, we consider an antipodal partition of $S^2$. In this way, we can select in any element of $\Qminion_\HH$ a distinguished coordinate --- namely, the coordinate whose corresponding subspace is mapped by the Hopf fibration to, say, the first part of the partition. The map assigning to each element of $\Qminion_\HH$ its distinguished coordinate is then shown to be a minion homomorphism from $\Qminion_\HH$ to $\Dminion$.\\
    $(ii)$ If $\HH$ has dimension at least $3$, $\Qminion_\HH$ is not equivalent to the dictator minion. Our proof of this fact works by showing that a homomorphism $\Qminion_\HH\to\Dminion$ would induce a Boolean measure on the subspaces of $\HH$. By virtue of Gleason's Theorem~\cite{gleason1975measures}, any measure $\mu$ on the subspaces of $\HH$ must be of the form $\mu(v)=\tr(m\cdot\pr_v)$ for some positive semidefinite linear operator $m$ (where $\pr_v$ denotes the orthogonal projector onto the space $v$). Thus, in particular, $\mu$ must have a continuous dependence on its inputs, which contradicts the fact that $\mu$ is Boolean.\\
    $(iii)$ For any $\HH$, $\Qminion_\HH$ admits a homomorphism to a linear minion $\Sminion$ introduced in~\cite{cz23soda:minions} and, independently, in~\cite{bgs_robust23stoc}, which captures the power of the \emph{standard semidefinite programming} relaxation for CSPs. The result is obtained by projecting a fixed vector of $\HH$ onto the spaces consitituing the elements of $\Qminion_\HH$. This yields a minion homomorphism from $\Qminion_\HH$ to a complex version $\Sminion_\C$ of $\Sminion$; then, we show that $\Sminion_\C$ and $\Sminion$ are in fact equivalent.\\[-1em]

The minion-theoretic results on $\Qminion_\HH$ listed above directly translate into results on the occurrence of quantum advantage. In particular, 
$(i)$ implies that quantum advantage can only happen over Hilbert spaces of dimension $3$ or higher --- an instance of the known phenomenon that strong non-locality cannot be realised by two-qubit systems, see~\cite{BrassardMT05,AbramskyBCSKM17}.
For such spaces, $(ii)$ and $(iii)$ yield new sufficient and necessary conditions for the occurrence of quantum advantage in terms of the complexity of the corresponding CSP:
$(ii)$ implies that, if $\CSP(\Y)$ is not decided by a polynomial-time algorithm, $\Y$ has quantum advantage;
$(iii)$ implies that, if $\Y$ has quantum advantage, $\CSP(\Y)$ has \emph{unbounded width}\footnote{We also describe an alternative derivation of this result, by showing that $\Qminion_\HH$ admits a homomorphism to the \emph{skeletal minion} introduced in~\cite{cz23sicomp:clap}, see Remark~\ref{rem_CLP_skeletal}. This minion corresponds to a stronger version of the \emph{singleton arc consistency} algorithm for CSPs~\cite{DB97,BD08,ChenDG13} and thus, similarly to the minion $\Sminion$ of semidefinite programming, it captures bounded-width CSPs~\cite{Kozik21:sicomp}.} --- where the width of a CSP is a central notion in the theory of CSP complexity, expressing the power of local-consistency techniques for its solution~\cite{Feder98:monotone,Barto14:jacm}. 

In the case of graphs, the combination of these results completely characterises the occurrence of quantum advantage: Using Hell--Ne\v{s}et\v{r}il's dichotomy for the complexity of graph CSPs (in Bulatov's reformulation~\cite{Bulatov05} in terms of polymorphisms), we show that, if $\dim(\HH)\geq 3$, a graph has quantum advantage if and only if it is  non-bipartite (while, if $\dim(\HH)\leq 2$, no graph has quantum advantage as stated above). This provides a complete answer to Question~\ref{question_quantum_advantage_graphs}.

Finally, we introduce a natural generalisation of quantum advantage. A structure $\Y$ having quantum advantage means that the presence of a quantum homomorphism from $\X$ to $\Y$ is not sufficient to guarantee the complete classical information $\X\to\Y$. Can it still provide some \emph{partial} classical information, in the form of a classical homomorphism $\X\to\Y'$ to a different structure $\Y'$? If the answer is negative, we say that the pair $(\Y,\Y')$ has quantum advantage.
We show that our minion framework is also able to capture this more general notion of quantum advantage. As a consequence, we connect quantum advantage for pairs of structures to the complexity of a \emph{promise} variant of CSPs formalised in~\cite{AGH17,BG21,BBKO21}.
In this way, we show that the two notions of quantum advantage do not collapse, in the sense that there exist two structures $\Y$ and $\Y'$ such that both $\Y$ and $\Y'$ have quantum advantage, but the pair $(\Y,\Y')$ does not.

\paragraph{Related work}

Understanding the power of entanglement as a resource in non-local games (like the homomorphism game described above) is a well-established research programme in quantum information and theoretical computer science. A prominent example of this line of work is the ``$\operatorname{MIP}^*=\operatorname{RE}$'' Theorem, recently proved in~\cite{JiNVWY21}, that characterises the class $\operatorname{MIP}^*$ of languages decided by a classical polynomial-time verifier interacting with multiple non-communicating computationally unbounded players sharing a finite-dimensional entangled state, as the class $\operatorname{RE}$ of recursively enumerable languages. In the same line of works, we mention~\cite{ItoV12,KempeKMTV11,Ito2008OracularizationAT,Natarajan019}.

Quantum perfect strategies in the special case of the graph $n$-colouring problem (i.e., when $\Y$ is the $n$-clique) and the corresponding notion of \emph{quantum chromatic number} were considered in~\cite{AvisHKS06,CameronMNSW07,fukawa2011quantum,MancinskaSS13,ScarpaS12,MancinskaR16,galliard2002pseudo,CleveHTW04}. 
Boolean CSPs (i.e., $\CSP(\Y)$ with $\Y$ being a Boolean relational structure) that are unsatisfiable but admit compatible assignments of self-adjoint unitary operators were characterised in~\cite{AtseriasKS19}. The relation between this notion of CSP satisfiability and quantum perfect strategies for non-local games was described in~\cite{CleveM14}, in the Boolean setting.

Unlike their non-promise variant, the complexity of promise CSPs remains largely undiscovered: Even the complexity of the promise version of graph $n$-colouring --- known as \emph{approximate graph colouring}~\cite{GJ76} --- is unknown in general.
Since the introduction of a minion-based algebraic approach to promise CSPs in~\cite{BBKO21}, the theory of minions has proved important in the investigation of this type of problems. For example, minions were used in~\cite{cz23sicomp:clap,cz23soda:minions,bgs_robust23stoc,bgwz20,dalmau2023local} to describe the power of relaxation techniques for promise CSPs, and in~\cite{KOWZ22,cz23stoc:ba,cz23soda:aip} to investigate the complexity and tractability under certain relaxation models of approximate graph colouring. 
Transferring the geometric techniques we develop in this work for describing the structure of the quantum minion
to other types of minions used in the theory of promise CSP complexity appears to be an interesting direction for future investigation.

Finally, we note that the dual version of Question~\ref{question_quantum_advantage_graphs} --- characterise the structures $\X$ for which there exists some $\Y$ such that $\X$ and $\Y$ are quantum but not classically homomorphic --- corresponds to the so-called \emph{left-hand side} CSP (LCSP), where $\X$ is a fixed parameter and $\Y$ is the input. For a single $\X$, LCSP is always trivially decidable in polynomial time. Hence, LCSP complexity questions are only studied in the setting where the left-hand side is an infinite class of structures, and their analysis is not captured by the algebraic approach~\cite{DalmauKV02,Grohe07}. Therefore, our methods are inapplicable for characterising this dual type of quantum advantage.

\section{Preliminaries}
\noindent By $\N$, $\R$, and $\C$ we denote the sets of positive integer numbers, real numbers, and complex numbers, respectively. For $\ell\in\N$, $[\ell]$ is the set $\{1,2,\dots,\ell\}$.
For $i\leq\ell\in\N$, $\be_{i;\ell}$ is the $i$-th standard unit vector of length $\ell$ (i.e., the $i$-th entry of $\be_{i;\ell}$ is $1$ and all other entries are $0$). 
Given $\ell\in\N$ and a set $S$, $S^\ell$ denotes the set of tuples of elements of $S$ of length $\ell$. 
For a complex matrix $M$, we denote the transpose and the conjugate transpose of $M$ by the symbols $M^\top$ and $M^*$, respectively.

\subsection{Finite-dimensional Hilbert spaces}
Throughout this work, we let $\HH$ denote a finite-dimensional Hilbert space (i.e., a finite-dimensional complex or real vector space equipped with an inner product $\ang{\cdot,\cdot}$).  
By $\dim(\HH)$ we denote the dimension of $\HH$, while $\bzero_\HH$ indicates the zero vector in $\HH$. To avoid trivial cases, we require $\HH\neq\{\bzero_\HH\}$.
Two subsets $v,w$ of $\HH$ are \emph{orthogonal} (in symbols, $v\perp w$) if $\ang{\bv,\bw}=0$ for each $\bv\in v$, $\bw\in w$. By $v^\perp$ we denote the \emph{orthogonal complement} of $v$; i.e., the set of all vectors $\bw\in\HH$ such that $\ang{\bv,\bw}=0$ for each $\bv\in v$. The set of all linear maps from $\HH$ to itself shall be denoted by $\End_\HH$.
Given $p,p'\in\End_\HH$, $p+p'$ denotes
their \emph{sum},
$pp'$ (or sometimes, for typographical convenience, $p\cdot p'$) denotes 
their \emph{product} or \emph{composition}, while 
$[p,p']=pp'-p'p$ is their \emph{commutator}.
The \emph{adjoint} of $p\in\End_\HH$ is the unique map $p^*\in\End_\HH$ satisfying $\ang{\bh,p^*(\bh')}=\ang{p(\bh),\bh'}$ for each $\bh,\bh'\in\HH$.
$p$ is \emph{self-adjoint} if $p^*=p$, and it is
an \emph{orthogonal projector} if it is self-adjoint and idempotent (i.e., if $p^2=p=p^*$). 
$\Proj_\HH$ denotes the set of all orthogonal projectors in $\End_\HH$, while $L_\HH$ denotes the set of all linear subspaces of $\HH$.
The function $\rg_{\bullet}$ sending a linear map $p\in\End_\HH$ to its range $\rg_p\subseteq\HH$ establishes a bijection between $\Proj_\HH$ and $L_\HH$. We denote the inverse of this bijection by $\pr_{\bullet}$ (i.e., for any $v\in L_\HH$, $\pr_v\in\Proj_\HH$ is the orthogonal projector onto $v$).
By $\id_\HH$ and $\zeroOp_{\HH}$ we denote the identity and zero maps on $\HH$, respectively.

\subsection{Relational structures}
A \emph{signature} $\sigma$ is a finite set of relation symbols $R$, each with its \emph{arity} $\ar(R)\in\N$.
A \emph{relational structure} $\X$ with signature $\sigma$ (in short, a $\sigma$-structure) consists of a set $X$ (the \emph{domain}\footnote{Unless otherwise stated, the domains of all relational structures appearing in this work shall be assumed to be finite.} of $\X$) and a relation $R^\X\subseteq X^{\ar(R)}$ for each symbol $R\in\sigma$. If $\sigma$ contains a unique symbol $R$ of arity $2$, and the relation $R^\X$ is symmetric and irreflexive (i.e., $(x,x')\in R^\X$ implies that $(x',x)\in R^\X$ and $x\neq x'$), $\X$ is called a (simple) \emph{graph}.
Two relational structures $\X,\Y$ are \emph{similar} if they have the same signature. In this case, a \emph{homomorphism} from $\X$ to $\Y$ is a map $f:X\to Y$ that preserves all relations; i.e., $f(\bx)\in R^{\Y}$ for each $\bx\in R^\X$, where $f$ is applied entrywise to the entries of $\bx$. We denote the existence of a homomorphism from $\X$ to $\Y$ by the expression $\X\to\Y$.

Given a $\sigma$-structure $\Y$, the \emph{constraint satisfaction problem} parameterised by $\Y$ (in short, $\CSP(\Y)$) is the following computational problem: Given as input a $\sigma$-structure $\X$, output $\YES$ if $\X\to\Y$, and $\NO$ if $\X\not\to\Y$.\footnote{This is the \emph{decision} version of $\CSP(\Y)$. In the \emph{search} version, the problem is to find an explicit homomorphism $f:\X\to\Y$ for any input $\X$ such that $\X\to\Y$. The two versions of $\CSP$ are equivalent up to polynomial-time reductions~\cite{BulatovJK05}.} 
For $\ell\in\N$, we denote by $\Y^\ell$ the $\ell$-th \emph{direct power} of $\Y$; i.e., $\Y^\ell$ is the $\sigma$-structure whose domain is $Y^\ell$ and whose relations are defined as follows: Given $R\in\sigma$  and any $\ell\times \ar(R)$ matrix $M$ whose rows are tuples in $R^\Y$, the columns of $M$ form a tuple in $R^{\Y^\ell}$. A \emph{polymorphism} of $\Y$ (of arity $\ell$) is a homomorphism from $\Y^\ell$ to $\Y$. By $\Pol(\Y)$ (resp. $\Pol^{(\ell)}(\Y)$) we denote the set of all polymorphisms (resp. all polymorphisms of arity $\ell$) of $\Y$.

\subsection{Quantum advantage}
\label{subsec_prelim_quantum_advantage}

Given two $\sigma$-structures $\X$ and $\Y$, 
consider a family $\{p_{x,y}\}_{x\in X,\,y\in Y}$ of orthogonal projectors in $\Proj_\HH$. For each $R\in\sigma$, $\bx\in X^{\ar(R)}$, and $\by\in Y^{\ar(R)}$, define $p_{\bx,\by}=p_{x_1,y_1}\cdot p_{x_2,y_2}\cdot\ldots\cdot p_{x_{\ar(R)},y_{\ar(R)}}$. Consider the following three conditions:
\begin{enumerate}
    \item[(Q1)] $\sum_{y\in Y}p_{x,y}=\id_\HH$ for each $x\in X$;
    \item[(Q2)] $[p_{x_i,y},p_{x_j,y'}]=\zeroOp_\HH$ for each $R\in\sigma$, $\bx\in R^\X$, $i,j\in [\ar(R)]$, and $y,y'\in Y$;
    \item[(Q3)] $p_{\bx,\by}=\zeroOp_\HH$ for each $R\in\sigma$, $\bx\in R^\X$, and $\by\in Y^{\ar(R)}\setminus R^\Y$.
\end{enumerate}
Following~\cite{abramsky2017quantum}, we say that $\X$ and $\Y$ are \emph{quantum homomorphic} over $\HH$ (and write $\X\qTo{\HH}\Y$) if there exists a family $\{p_{x,y}\}_{x\in X,\,y\in Y}$ having the three properties listed above. It was shown in~\cite{abramsky2017quantum} that $\X\qTo{\HH}\Y$ if and only if there exists a \emph{quantum perfect strategy} for the $\X\Y$-homomorphism game mentioned in the Introduction. Observe that $\X\to\Y$ always implies $\X\qTo{\HH}\Y$. 
Indeed, if $f:\X\to\Y$ is a homomorphism, letting $p_{x,y}=\id_\HH$ if $y=f(x)$ and $p_{x,y}=\zeroOp_\HH$ otherwise yields a proper quantum homomorphism.
We say that $\Y$ \emph{has quantum advantage over $\HH$} if the converse is not true; i.e., if there exists a $\sigma$-structure $\X$ such that $\X\qTo{\HH}\Y$ but $\X\not\to\Y$. In terms of the $\X\Y$-homomorphism game, this means that Alice and Bob can win the game using a quantum-assisted strategy, but they cannot using a classical strategy.\footnote{We point out that, if $\X$ and $\Y$ are graphs, the definition of quantum homomorphism from~\cite{abramsky2017quantum} is more restrictive than the one given in~\cite{MancinskaRoberson16}, in that the latter does not enforce the commutativity condition (Q2) among the projectors. Nevertheless, in the case of graphs, our results hold for both notions of quantum homomorphisms, see Remark~\ref{rem_MR_homomorphisms}.}
\subsection{Minions}
A minion $\Mminion$ is the disjoint union of non-empty sets $\Mminion^{(\ell)}$ for $\ell\in\N$, equipped with operations $(\cdot)_{/\pi}:\Mminion^{(\ell)}\to \Mminion^{(\ell')}$ (known as \emph{minor operations}) for each pair of integers $\ell,\ell'\in\N$ and each map
$\pi:[\ell]\to [\ell']$, which must satisfy the following two conditions:
\begin{itemize}
    \item $M_{/\id}=M$ 
    for each $\ell\in\N$ and $M\in\Mminion^{(\ell)}$, where $\id$ is the identity map on the set $[\ell]$;
    \item $(M_{/\pi})_{/\tilde\pi}=M_{/\tilde\pi\circ\pi}$
     for each $\ell,\ell',\ell''\in\N$, $M\in\Mminion^{(\ell)}$, $\pi:[\ell]\to[\ell']$,
    and $\tilde\pi:[\ell']\to[\ell'']$.\footnote{Equivalently, we may define a minion as a functor from the skeleton category of non-empty finite sets to the category of non-empty sets. The definition of minion as we present it here was introduced in~\cite{bgwz20}, and it is an abstraction of the \emph{minor-closed classes of functions} considered in~\cite{Pippenger02} to study a Galois connection between functions and relations. In fact, the latter objects are precisely the function minions of Example~\ref{example_function_minion}.}
\end{itemize}
Given $M\in\Mminion$, the index $\ell$ for which $M\in\Mminion^{(\ell)}$ is referred to as the \emph{arity} of $M$. A \emph{subminion} of a minion $\Mminion$ is a non-empty subset of $\Mminion$ that is closed under the minor operations of $\Mminion$.
In this work, we shall consider two concrete types of minions, described in the two examples below.

\begin{example}[Function minions]
\label{example_function_minion}
Take two sets $S$ and $T$. For any integer $\ell\in\N$, consider the set $\mathscr{F}^{(\ell)}_{S,T}$ of all functions from $S^\ell$ to $T$. Given $f\in\mathscr{F}^{(\ell)}_{S,T}$ and $\pi:[\ell]\to[\ell']$, we let $f_{/\pi}\in\mathscr{F}^{(\ell')}_{S,T}$ be the function defined as follows:
\begin{align*}
f_{/\pi}:S^{\ell'}&\to T\\
    (s_1,\dots,s_{\ell'})&\mapsto f(s_{\pi(1)},\dots,s_{\pi(\ell)}).
\end{align*}
The set $\mathscr{F}_{S,T}$ consisting in the disjoint union of all sets $\mathscr{F}^{(\ell)}_{S,T}$ for $\ell\in\N$ is easily seen to be a minion. A \emph{function minion} (over the sets $S$ and $T$) is any subminion of $\mathscr{F}_{S,T}$.
\end{example}
\begin{example}[Linear minions]
\label{example_linear_minions}
Let $\mathcal{A}$ be an abelian monoid --- i.e., a set equipped with a binary operation ``$+_\mathcal{A}$'' that is associative and commutative, and admits an identity element ``$0_\mathcal{A}$''.
Following~\cite{cz23soda:minions}, we let a \emph{linear minion} over $\mathcal{A}$ be a non-empty set $\Mminion$ of matrices having entries in $\mathcal{A}$ that is closed under the following elementary row operations:
\begin{itemize}
    \item swapping two rows;
    \item replacing two rows with their sum;
    \item inserting an extra zero row.
\end{itemize}
(Here, the sum is the entrywise application of the ``$+_\mathcal{A}$'' operation, while the zero row is a tuple of copies of ``$0_{\mathcal{A}}$''.) We let $\Mminion^{(\ell)}$ be the set of all matrices in $\Mminion$ having exactly $\ell$-many rows. Moreover, given a matrix $M\in\Mminion^{(\ell)}$ and a map $\pi:[\ell]\to[\ell']$, we define $M_{/\pi}$ as the matrix whose $i$-th row, for $i\in[\ell']$, is the sum of all rows of $M$ having index in $\pi^{-1}(i)$; i.e., 
\begin{align*}
\be_{i;\ell'}^{\top}M_{/\pi}
    =
    \sum_{j\in\pi^{-1}(i)}\be_{j;\ell}^{\top}M.
\end{align*}
Starting from $M$, one can realise $M_{/\pi}$ through a sequence of the elementary row operations listed above. Therefore, $M_{/\pi}\in\Mminion^{(\ell')}$. Furthermore, this choice for the minor operations guarantees that $\Mminion$ is a minion.
\end{example}

Let $\Mminion$ and $\Nminion$ be two minions. A \emph{minion homomorphism} $\xi:\Mminion\to\Nminion$ is a map from the underlying set of $\Mminion$ to the underlying set of $\Nminion$ that
\begin{itemize}
    \item preserves the arity --- i.e., $\xi(M)\in\Nminion^{(\ell)}$ if $M\in\Mminion^{(\ell)}$;
    \item preserves the minors --- i.e., $\xi(M_{/\pi})=\xi(M)_{/\pi}$ if $M\in\Mminion^{(\ell)}$ and  $\pi:[\ell]\to[\ell']$.\footnote{Equivalently, a minion homomorphism $\xi:\Mminion\to\Nminion$
    is a natural transformation from $\Mminion$ to $\Nminion$.}
\end{itemize}
Given a relational structure $\Y$, the set $\Pol(\Y)$ of all polymorphisms of $\Y$ is a subset of $\mathscr{F}_{Y,Y}$ (as defined in Example~\ref{example_function_minion}), it is non-empty (since it contains the identity homomorphism), and it is closed under the minor operations of $\mathscr{F}_{Y,Y}$; hence, it is a (function) minion.
Polymorphism minions capture the complexity of CSPs, in the sense of the following result.
\begin{thm}[\cite{BOP18}]
\label{thm_minion_homos_reductions}
    Let $\Y$ and $\Y'$ be two relational structures. If there exists a minion homomorphism $\Pol(\Y)\to\Pol(\Y')$, then $\CSP(\Y')$ reduces in polynomial time to $\CSP(\Y)$.\footnote{In fact, a log-space reduction from $\CSP(\Y')$ to $\CSP(\Y)$ exists.}
\end{thm}

\section{The quantum minion}
\label{sec_quantum_minion}
In this section, we describe a new minion that captures quantum advantage. The geometric properties of this minion 
shall be used in the later sections to obtain information on the occurrence of quantum advantage.

Recall that $L_\HH$ denotes the set of all linear subspaces of $\HH$. Given $v,w\in L_\HH$, we let $v\bplus w$ denote their sum as linear spaces; i.e.,  $v\bplus w=\{\bv+\bw\mid\bv\in v,\bw\in w\}$. Let $\mathcal L_\HH$ denote the set $L_\HH$ equipped with the operation ``$\bplus$'', and observe that $\mathcal L_\HH$ is an abelian monoid, with the subspace $\{\bzero_\HH\}$ as the identity element. 
For $\ell\in\N$, consider the set 
\begin{align*}
    \Qminion_\HH^{(\ell)}
    =
    \left\{
    \begin{array}{llll}
         \bq=
\begin{bmatrix}
q_1\\\vdots\\q_\ell
\end{bmatrix}
 \;\;\;\mbox{such that}\;\;\;
\begin{array}{lll}
     q_i\in L_\HH \mbox{ for each } i,  \\[5pt]
    q_i\perp q_j \mbox{ for each } i\neq j,\\[5pt]
    \displaystyle\bigboxplus_{i\in [\ell]}q_i=\HH
\end{array} 
    \end{array}
    \right\}.
\end{align*}
We define the \emph{quantum minion} $\Qminion_\HH$ as the disjoint union of the sets $\Qminion_\HH^{(\ell)}$ for $\ell\in\N$.
\begin{lem}
    $\Qminion_\HH$ is a linear minion over the abelian monoid $\mathcal L_\HH$.\footnote{In fact, $\Qminion_\HH$ belongs to a particular class of linear minions known as \emph{conic minions}, introduced in~\cite{cz23soda:minions} (see also~\cite{dalmau2023local}).}
\end{lem}
\begin{proof}
All elements of $\Qminion_\HH$ are vectors (i.e., matrices having a single column) of spaces in $L_\HH$. The statement immediately follows by observing that $\Qminion_\HH$ is closed under the three elementary row operations of Example~\ref{example_linear_minions}.
\end{proof}
Using the quantum minion, we will show that quantum advantage is a property of the polymorphism minion of relational structures. Formally, 
we will prove the next result.
\begin{thm}
\label{thm_quantum_trivial_iff_minion_homo}
    A relational structure $\Y$ has quantum advantage over $\HH$ if and only if $\Qminion_\HH\not\to\Pol(\Y)$.
\end{thm}
The proof of Theorem~\ref{thm_quantum_trivial_iff_minion_homo} makes use of two ingredients: the \emph{free structures} from~\cite{BBKO21} and the \emph{minion tests} from~\cite{cz23soda:minions}.

Let $\Mminion$ be a minion and let $\Y$ be a $\sigma$-structure. 
On a high level, the free structure $\freeM(\Y)$ of $\Mminion$ generated by $\Y$ is obtained by simulating the relations of $\Y$ on a domain consisting of elements of
$\Mminion$. Formally, following~\cite{BBKO21}, we let $\freeM(\Y)$ be the (potentially infinite) $\sigma$-structure whose domain is the set $\Mminion^{(n)}$ (where $n$ is the domain size of $\Y$) and whose relations are described as follows. For each symbol $R\in\sigma$, 
let $m$ be the cardinality of $R^\Y$.
A tuple $(M_1,\dots,M_{\ar(R)})$ of elements
of $\Mminion^{(n)}$ belongs to $R^{\freeM(\Y)}$ if and only if there exists some
$\tilde M\in \Mminion^{(m)}$ such that $M_i=\tilde M_{/\pi_{i}}$ for each $i\in[\ar(R)]$, where
$\pi_{i}:[m]\to[n]$ is the function mapping each tuple $\by\in R^\Y$ to its $i$-th entry $y_i$.\footnote{Here we are implicitly identifying $Y$ with $[n]$ and $R^\Y$ with $[m]$, in order to avoid introducing extra notation for the bijections between those sets.}
Notice that the domain of $\freeM(\Y)$ can be infinite even when the domain of $\Y$ is finite. In particular, this is the case for $\Mminion=\Qminion_\HH$ (except in the cases where the dimension of $\HH$ or the domain size of $\Y$ equals $1$).

The next result, whose proof is deferred to Appendix~\ref{appendix_omitted_proofs}, shows that quantum homomorphisms correspond to classical homomorphisms to the free structure of the quantum minion.\footnote{The idea of expressing quantum homomorphisms as classical homomorphisms to a modified structure already appeared in~\cite{MancinskaRoberson16,abramsky2017quantum}. More precisely,
in the case of graphs, the structure $\freeQH(\Y)$ of Proposition~\ref{prop_quantum_homo_as_test} corresponds to the graph $M(\Y,d)$ from~\cite[Theorem~2.9]{MancinskaRoberson16}, where $d=\dim(\HH)$. Moreover, for arbitrary relational structures, it corresponds to the structure $\mathcal{Q}_d\Y$, where $\mathcal{Q}_d$ is the \emph{quantum monad} introduced in~\cite{abramsky2017quantum}.
In our setting, however, it is crucial to formulate Proposition~\ref{prop_quantum_homo_as_test} in terms of the free structure of a minion.
This allows us to use the notion of minion tests and, through Theorem~\ref{thm_minion_test_solvability_homo}, to express quantum advantage as a property of the polymorphism minion (Theorem~\ref{thm_quantum_trivial_iff_minion_homo}), thus linking it to the complexity of the corresponding CSP.}
\begin{prop}
\label{prop_quantum_homo_as_test}
    Let $\X,\Y$ be similar relational structures. Then $\X\qTo{\HH}\Y$ if and only if $\X\to\freeQH(\Y)$.
\end{prop}

Next, we describe the notion of minion tests from~\cite{cz23soda:minions}. Let $\Mminion$ be a minion, let $\sigma$ be a signature, and let $\Sigma$ be the class of all (finite) $\sigma$-structures. Consider the function 
\begin{align*}
    \Test{\Mminion}{}:\Sigma^2&\to\{\YES,\NO\}\\
    (\X,\Y)&\mapsto \left\{\begin{array}{lll}
         \YES&\mbox{if }\;\X\to\freeM(\Y)  \\
         \NO&\mbox{otherwise.} 
    \end{array}\right.
\end{align*}
Fix a $\sigma$-structure $\Y$. It is not hard to verify that $\Y\to\freeM(\Y)$ for every $\Mminion$. Hence, $\X\to\Y$ implies $\Test{\Mminion}{}(\X,\Y)=\YES$. We say that $\Test{\Mminion}{}$ \emph{solves} $\CSP(\Y)$ if the converse implication also holds; i.e., if, for every $\sigma$-structure $\X$, $\Test{\Mminion}{}(\X,\Y)=\YES$ implies $\X\to\Y$.
As stated next, this can be reformulated in terms of minion homomorphisms.
\begin{thm}[\cite{cz23soda:minions}]
\label{thm_minion_test_solvability_homo}
    Let $\Mminion$ be a minion and let $\Y$ be a relational structure. Then $\Test{\Mminion}{}$ solves $\CSP(\Y)$ if and only if $\Mminion\to\Pol(\Y)$.
\end{thm}
We can now prove Theorem~\ref{thm_quantum_trivial_iff_minion_homo}.
\begin{proof}[Proof of Theorem~\ref{thm_quantum_trivial_iff_minion_homo}]
By virtue of Proposition~\ref{prop_quantum_homo_as_test}, the existence of a quantum homomorphism over $\HH$ can be viewed as a minion test for the quantum minion $\Qminion_\HH$.
More precisely, given two $\sigma$-structures $\X$ and $\Y$, $\X\qTo{\HH}\Y$ if and only if $\Test{\Qminion_\HH}{}(\X,\Y)=\YES$.
Therefore, $\Y$ has quantum advantage over $\HH$ precisely when $\Test{\Qminion_\HH}{}$ does not solve $\CSP(\Y)$. The result then follows from Theorem~\ref{thm_minion_test_solvability_homo}.   
\end{proof}
\section{$\Qminion_\HH$ and the dictator minion}
By virtue of Theorem~\ref{thm_quantum_trivial_iff_minion_homo}, in order to investigate the occurrence of quantum advantage, one needs to understand the structure of the quantum minion $\Qminion_\HH$. ``Minion homomorphism'' is a binary, reflexive, and transitive relation; thus, it induces a preorder on the class of all minions. It shall be crucial to obtain information on where $\Qminion_\HH$ is located within this preorder. We shall see that, when $\dim(\HH)\leq 2$, $\Qminion_\HH$ is a \emph{least element} of the preorder, as it is equivalent to a specific minion --- known as the dictator minion --- that admits a homomorphism to all minions.
This implies that no structure has quantum advantage over such Hilbert spaces.
Things get more interesting when
$\dim(\HH)\geq 3$. In this case, we shall see that $\Qminion_\HH$ is essentially different from the dictator minion. 
Through known results on the complexity of CSPs and the theory of their polymorphism minions, this separation shall lead to a sufficient condition for the occurrence of quantum advantage.
We start by defining the \emph{dictator minion} as the linear minion $\Dminion$ over the abelian monoid $(\R,+)$ consisting of all standard unit vectors $\be_{i;\ell}$ for $1\leq i\leq\ell\in\N$.\footnote{Equivalently, we may define $\Dminion$ as a function minion over some finite set with at least two elements, consisting of functions that project any input tuple onto one fixed coordinate --- whence the alternative name  \emph{projection minion}, which is common in the CSP literature (see~\cite{BBKO21}). Yet another description of $\Dminion$ is the following: For $\ell\in\N$, $\Dminion^{(\ell)}=[\ell]$; for $\pi:[\ell]\to[\ell']$ and $i\in\Dminion^{(\ell)}$, $i_{/\pi}=\pi(i)$.} 
The next simple result shows that $\Dminion$ is a least element of the minion homomorphism preorder.\footnote{The minion homomorphism preorder also has a \emph{greatest element}: The linear minion (over any abelian monoid) whose only $\ell$-ary element, for each $\ell\in\N$, is the all-zero vector of length $\ell$.} 
\begin{lem}
\label{lem_projections_map_everywhere}
    $\Dminion\to\Mminion$ for every minion $\Mminion$.
\end{lem}
\begin{proof}
    Fix an element $M\in\Mminion^{(1)}$ (which must exist as minions are non-empty by definition). For $i\leq\ell\in\N$, consider the map $\pi_{i,\ell}:[1]\to[\ell]$ given by $1\mapsto i$. It is straightforward to check that the function $\be_{i;\ell}\mapsto M_{/\pi_{i,\ell}}$ yields a minion homomorphism from $\Dminion$ to $\Mminion$.
\end{proof}

We shall establish the following characterisation of when the quantum minion admits a homomorphism (and, thus, is homomorphically equivalent) to the dictator minion.
\begin{thm}
\label{thm_quantum_minion_vs_projections}
$\Qminion_\HH\to\Dminion$ if and only if $\dim(\HH)\leq 2$.
\end{thm}

Our proofs of the two directions of Theorem~\ref{thm_quantum_minion_vs_projections} are both geometric. This is not surprising, as the quantum minion $\Qminion_\HH$ has a natural geometric interpretation. 

The ``if'' part (Proposition~\ref{prop_if_part_thm_dictator_quantum}) is established as follows. First, we show that every $2$-dimensional Hilbert space can be partitioned into two parts in a way that any pair of nonzero orthogonal vectors in the space is separated by the partition. In order to prove this fact, we consider the Hopf fibration of the complex projective line $\C\textbf{P}^1$ into the $2$-sphere\footnote{This construction results in the object that is sometimes called the \emph{Riemann sphere} or \emph{Bloch sphere}.} --- which has the property of mapping orthogonal vectors to antipodal points of the sphere. We can then partition the Hilbert space by taking the preimage under the Hopf fibration of an antipodal partition of the $2$-sphere.
\begin{lem}
\label{lem_Hopft_fibration}
    If $\dim(\HH)=2$, there exists $C\subset \HH$ such that, given any two nonzero orthogonal
    vectors $\bv,\bw\in\HH$, exactly one of $\bv$ and $\bw$ lies in $C$.
\end{lem}
A self-contained proof of Lemma~\ref{lem_Hopft_fibration} is given in Appendix~\ref{appendix_omitted_proofs}.
Once such a partition $\{C,\HH\setminus C\}$ is fixed, we can select in each pair of $1$-dimensional orthogonal subspaces of $\HH$ a distinguished element --- namely, the subspace generated by a vector in $C$. As we shall see next, the property of the partition guarantees that the distinguished element is preserved under taking minors, thus yielding a proper homomorphism from $\Qminion_\HH$ to $\Dminion$.

\begin{prop}
\label{prop_if_part_thm_dictator_quantum}
    If $\dim(\HH)\leq 2$, $\Qminion_\HH\to\Dminion$.
\end{prop}
\begin{proof}
Given $\bq\in\Qminion_\HH$ of some arity $\ell$, an \emph{essential coordinate} of $\bq$ is an index $i\in[\ell]$ for which $q_i\neq\{\bzero_\HH\}$. We denote the set of essential coordinates of $\bq$ by $\operatorname{sec(\bq)}$.
Note that the cardinality of $\sec(\bq)$ is at most $\dim(\HH)\leq 2$. 
For any $1$-dimensional subspace $v\in L_\HH$, fix a nonzero vector $g(v)\in v$.
Let also $C\subset \HH$ be the set constructed in Lemma~\ref{lem_Hopft_fibration}.

We define $\xi:\Qminion_\HH\to\Dminion$ as follows. Given $\bq\in\Qminion_\HH$ of arity $\ell$, if $\bq$ has a unique essential coordinate $i$, we let $\xi(\bq)=\be_{i;\ell}$. Otherwise, $\bq$ has two essential coordinates $i\neq j$ --- in which case, it must be that $\dim(\HH)=2$ and $\dim(q_i)=\dim(q_j)=1$.
If $g(q_i)\in C$, we let $\xi(\bq)=\be_{i;\ell}$; otherwise, we let $\xi(\bq)=\be_{j;\ell}$. Since $\ang{g(q_i), g(q_j)}=0$, $\xi$ is well defined by virtue of the property of $C$.
We claim that $\xi$ yields a minion homomorphism from $\Qminion_\HH$ to $\Dminion$.

The facts that $\xi(\bq)\in\Dminion$ for each $\bq\in\Qminion_\HH$ and that $\xi$ preserves the arity directly follow from the definition. 
To show that $\xi$ preserves the minors, take $\pi:[\ell]\to[\ell']$ and $\bq\in\Qminion_\HH^{(\ell)}$. Observe that 
\begin{align*}
    \sec(\bq_{/\pi})=\{\pi(i):i\in\sec(\bq)\}.
\end{align*}
If $\sec(\bq)=\{i\}$, it follows that 
\begin{align*}
   \xi(\bq_{/\pi})=\be_{\pi(i);\ell'}=(\be_{i;\ell})_{/\pi}=\xi(\bq)_{/\pi}.
\end{align*}
If $\sec(\bq)=\{i,j\}$ with $i\neq j$, let $k\in\{i,j\}$ be such that $\xi(\bq)=\be_{k;\ell}$. There are two cases. If $\pi(i)=\pi(j)$, we have that $\sec(\bq_{/\pi})=\{\pi(k)\}$, and it follows that
$\xi(\bq_{/\pi})=\be_{\pi(k);\ell'}$.
If $\pi(i)\neq\pi(j)$, $\bq_{/\pi}$ is the vector in $\Qminion_\HH^{(\ell')}$ whose $\pi(i)$-th entry is $q_i$, whose $\pi(j)$-th entry is $q_j$, and all of whose other entries are $\{\bzero_\HH\}$. Hence, also in this case it holds that $\xi(\bq_{/\pi})=\be_{\pi(k);\ell'}$.
This means that, in either case, we have
\begin{align*}
    \xi(\bq_{/\pi})=\be_{\pi(k);\ell'}=(\be_{k;\ell})_{/\pi}=\xi(\bq)_{/\pi}.
\end{align*}
Therefore, the identity $\xi(\bq_{/\pi})=\xi(\bq)_{/\pi}$ is satisfied in all cases, and it follows that $\xi$ is a minion homomorphism.
\end{proof}

As a direct consequence, we obtain the following.
\begin{cor}
\label{cor_no_quantum_advantage_in_low_dimension}
    No relational structure has quantum advantage over $\HH$ if $\dim(\HH)\leq 2$. 
\end{cor}
\begin{proof}
    If $\dim(\HH)\leq 2$,  Lemma~\ref{lem_projections_map_everywhere} and Proposition~\ref{prop_if_part_thm_dictator_quantum} yield
    \begin{align*}
        \Qminion_\HH\to\Dminion\to\Pol(\Y)
    \end{align*}
    for any relational structure $\Y$, and the conclusion immediately follows from Theorem~\ref{thm_quantum_trivial_iff_minion_homo}.
\end{proof}
Corollary~\ref{cor_no_quantum_advantage_in_low_dimension} is an instance of a more general phenomenon in quantum information, observed in~\cite{BrassardMT05} in the context of pseudo-telepathy games (see also~\cite{AbramskyBCSKM17}): In order for an $n$-qubit system to exhibit strong non-locality with any finite number of local measurements, $n$ must be at least $3$.
We now turn to the ``only if'' part of Theorem~\ref{thm_quantum_minion_vs_projections} (Proposition~\ref{prop_only_if_part_thm_dictator_quantum}). Our proof builds on a deep result by Gleason on measures on the subspaces of Hilbert spaces ---  which we state next, after introducing the necessary terminology.

The \emph{trace} $\tr(p)$ of a linear map $p\in\End_\HH$ is the trace of a matrix representation of $p$ in any basis of $\HH$.
We say that $p$ is \emph{positive semidefinite} (and write $p\succeq \zeroOp_\HH$) if $p$ is self-adjoint and $\ang{p(\bh),\bh}\geq 0$ for each $\bh\in\HH$.
A \emph{measure on the subspaces of $\HH$} is a function $\mu:L_\HH\to\R_{\geq 0}$ such that, for each collection $\{v_1,\dots,v_\ell\}$ of mutually orthogonal subspaces of $\HH$, it holds that
\begin{align}
\label{eqn_12_06_13012024}
  \mu\left(\bigboxplus_{i\in [\ell]}v_i\right)=\sum_{i\in [\ell]}\mu(v_i).  
\end{align}

\begin{thm}[\cite{gleason1975measures}]
\label{thm_gleason}
    Let $\mu$ be a measure on the subspaces of a finite-dimensional (real or complex) Hilbert space $\HH$ of dimension at least three. Then there exists a positive semidefinite linear map $m\in\End_\HH$ such that 
    \begin{align}
    \label{eqn_1301_1620}
        \mu(v)=\tr(m\cdot\pr_v)
    \end{align} 
    for each $v\in L_\HH$.\footnote{We point out that Gleason's Theorem holds in the more general setting of separable, not necessarily finite-dimensional Hilbert spaces $\HH$. In the general formulation, $\mu$ needs to be defined on all closed subspaces of $\HH$, and it needs to satisfy the condition~\eqref{eqn_12_06_13012024} for each countable collection of mutually orthogonal subspaces whose linear span is closed.}
\end{thm}

Proposition~\ref{prop_only_if_part_thm_dictator_quantum} is established by showing that a homomorphism from the quantum minion to the dictator minion would yield a well-defined, Boolean measure $\mu$ on the subspaces of $\HH$. However, if the dimension of $\HH$ is at least $3$, 
Theorem~\ref{thm_gleason} implies that $\mu$ must have a continuous dependence on its inputs and, in particular, it must assume intermediate values between $0$ and $1$, thus yielding a contradiction.

\begin{prop}
\label{prop_only_if_part_thm_dictator_quantum}
    If $\dim(\HH)\geq 3$, $\Qminion_\HH\not\to\Dminion$.
\end{prop}
\begin{proof}
    Suppose, for the sake of contradiction, that there exists a homomorphism $\xi:\Qminion_\HH\to\Dminion$. We associate with any linear subspace $v\in L_\HH$ the vector
    \begin{align*}
        \widehat{v}
        =
        \begin{bmatrix}
        v\\v^\perp
    \end{bmatrix}\in\Qminion_\HH^{(2)}.  
    \end{align*}
Recall that $\Dminion^{(2)}=\{\be_{1;2},\be_{2;2}\}$,
and consider the function $\mu: L_\HH\to\R_{\geq 0}$ defined by 
    \begin{align*}
        \mu(v)
        =
        \left\{
        \begin{array}{llll}
             1&\mbox{if}&\xi(\widehat{v})=\be_{1;2}  \\
             0&\mbox{if}&\xi(\widehat{v})=\be_{2;2}
        \end{array}
        \right.
    \end{align*}
for each $v\in L_\HH$.
We claim that $\mu$ is a measure on the subspaces of $\HH$.
Take a collection $\{v_1,\dots,v_\ell\}$ of mutually orthogonal subspaces of $\HH$, and let $w=\bigboxplus_{i\in [\ell]}v_i$. 

Suppose first that $\mu(w)=1$, which means that $\xi(\widehat{w})=\be_{1;2}$. Take the vector
\begin{align*}
    \bq=
    \begin{bmatrix}
        v_1&
        v_2&
        \dots&
        v_\ell&
        w^\perp
    \end{bmatrix}^{\top}
    \in\Qminion_\HH^{(\ell+1)},
    \end{align*}
and consider the map $\pi:[\ell+1]\to [2]$ such that $\pi(i)=1$ if $i\in [\ell]$, and $\pi(\ell+1)=2$. Note that $\bq_{/\pi}=\widehat{w}$; using that $\xi$ is a homomorphism, we get
\begin{align}
\label{eqn_12182023}
    \be_{1;2}
    =
    \xi(\widehat{w})
    =
    \xi(\bq_{/\pi})
    =
    \xi(\bq)_{/\pi},
\end{align}
whence it follows that $\xi(\bq)=\be_{j;\ell+1}$ for some $j\in [\ell]$. Consider now, for $k\in [\ell]$, the map $\pi_k:[\ell+1]\to[2]$ such that $\pi_k(k)=1$ and $\pi_k(i)=2$ for each $i\in [\ell+1]\setminus\{k\}$.
Observe that $\bq_{/\pi_k}=\widehat{v_k}$. Therefore,
\begin{align}
\label{eqn_1221_2023}
    \xi(\widehat{v_k})
    =
    \xi(\bq_{/\pi_k})
    =
    \xi(\bq)_{/\pi_k}
    =
    (\be_{j;\ell+1})_{/\pi_k}
    =
    \be_{\pi_k(j);2}.
\end{align}
It follows from the above that $\mu(v_k)=1$ precisely when $\pi_k(j)=1$. In other words, $\mu(v_j)=1$ and $\mu(v_k)=0$ for each $k\neq j$.
Therefore,
\begin{align*}
    \sum_{i\in [\ell]}\mu(v_i)
    =
    1
    =
    \mu(w)
    =
    \mu\left(\bigboxplus_{i\in [\ell]}v_i\right),
\end{align*}
as needed. 

Suppose now that $\mu(w)=0$. In this case, the equation~\eqref{eqn_12182023} turns into
\begin{align*}
    \be_{2;2}
    =
    \xi(\widehat{w})
    =
    \xi(\bq_{/\pi})
    =
    \xi(\bq)_{/\pi},
\end{align*}
which means that $\xi(\bq)=\be_{\ell+1;\ell+1}$.
Then, for each $k\in [\ell]$, the equation~\eqref{eqn_1221_2023} becomes
\begin{align*}
    \xi(\widehat{v_k})
    =
    \xi(\bq_{/\pi_k})
    =
    \xi(\bq)_{/\pi_k}
    =
    (\be_{\ell+1;\ell+1})_{/\pi_k}
    =
    \be_{\pi_k(\ell+1);2}
    =
    \be_{2;2},
\end{align*}
thus showing that $\mu(v_k)=0$ for each $k\in [\ell]$. It follows that 
\begin{align*}
    \sum_{i\in [\ell]}\mu(v_i)
    =
    0
    =
    \mu(w)
    =
    \mu\left(\bigboxplus_{i\in [\ell]}v_i\right).
\end{align*}

Hence, $\mu$ is a measure on the subspaces of $\HH$, as claimed. We can then invoke Theorem~\ref{thm_gleason} and find a linear map $m\in\End_\HH$ such that $m\succeq \zeroOp_\HH$ and $\mu(v)=\tr(m\cdot \pr_v)$ for each $v\in L_\HH$. 
In particular, 
\begin{align}
\label{eqn_1301_1700}
    \mu(\HH)
    =\tr(m\cdot\pr_\HH)
    =\tr(m\cdot\id_\HH)
    =\tr(m).
\end{align}
We now claim that $\mu(\HH)=1$. Indeed, if $\mu(\HH)=0$, it would follow from~\eqref{eqn_1301_1700} that 
$\tr(m)=0$; since $m$ is positive semidefinite, this would mean that $m=\zeroOp_\HH$; then,~\eqref{eqn_1301_1620} would force $\mu(v)=0$ for each $v\in L_\HH$. On the other hand, letting $\tau:[2]\to[2]$ be the map $1\mapsto 2$, $2\mapsto 1$, we would find
\begin{align*}
    \xi(\widehat{\{\bzero_\HH\}})
    =
    \xi(\widehat{\HH}_{/\tau})
    =
    \xi(\widehat{\HH})_{/\tau}
    =
    ({\be_{2;2}})_{/\tau}
    =
    \be_{1;2}
\end{align*}
and, thus, $\mu(\{\bzero_\HH\})=1$, a contradiction.

Next, let $d=\dim(\HH)$, let $\{\bv_1,\dots,\bv_d\}$ be an orthonormal basis of eigenvectors for $m$, and
let $\lambda_1,\dots,\lambda_d$ be the corresponding (real, nonnegative) eigenvalues of $m$. It follows from~\eqref{eqn_1301_1700} that 
\begin{align*}
    1=\mu(\HH)=\tr(m)=\sum_{i\in [d]}\lambda_i.
\end{align*}
Therefore, letting $v_i$ be the linear span of $\bv_i$, we have
\begin{align*}
    \mu(v_i)
    =
    \tr(m\cdot\pr_{v_i})
    =
    \lambda_i
\end{align*}
for each $i\in[d]$.
On the other hand, the range of $\mu$ is included in $\{0,1\}$ by construction. We deduce that the spectrum of $m$ consists of the values $1$ (with multiplicity $1$) and $0$ (with multiplicity $d-1$). In other words, $m=\pr_{v_j}$ for some $j\in [d]$. Choose an index ${k}\neq j$, and let $w$ be the linear span of the vector $\bv_j+\bv_{k}$. Using that $m$ is self-adjoint, we obtain
\begin{align*}
    \mu(w)
    &=
    \tr(m\cdot\pr_w)
    =
    \sum_{i\in [d]}
    \ang{m\cdot\pr_w(\bv_i),\bv_i}\\
    &=
    \sum_{i\in [d]}
    \ang{\pr_w(\bv_i),m(\bv_i)}
    =
    \sum_{i\in [d]}
    \ang{\pr_w(\bv_i),\pr_{v_j}(\bv_i)}\\
    &=
    \ang{\pr_w(\bv_j),\bv_j}.
\end{align*}
Since an orthonormal basis for $w$ is given by the single vector $\bw=\frac{1}{\sqrt{2}}(\bv_j+\bv_{k})$, we have that 
\begin{align*}
    \pr_w(\bv)
    =
    \ang{\bv,\bw}\bw
\end{align*}
for each $\bv\in\HH$. Observe that $\ang{\bw,\bv_j}=\frac{1}{\sqrt{2}}$. As a consequence,
\begin{align*}
    \mu(w)
    =
    \left\langle\ang{\bv_j,\bw}\bw,\bv_j\right\rangle
    =
    \frac{1}{2},
\end{align*}
a contradiction.
\end{proof}
Proposition~\ref{prop_only_if_part_thm_dictator_quantum} yields a separation of the quantum minion from the dictator minion. Together with Theorem~\ref{thm_quantum_trivial_iff_minion_homo}, it implies that $\Y$ has quantum advantage over any Hilbert space of dimension at least $3$ whenever $\Pol(\Y)\to\Dminion$.
In turn, the algebraic theory of CSPs
allows formulating the latter condition in terms of the complexity of $\CSP(\Y)$.

Since homomorphisms between polymorphism minions provide polynomial-time reductions between CSPs (Theorem~\ref{thm_minion_homos_reductions}) and since the dictator minion is homomorphic to any minion (Lemma~\ref{lem_projections_map_everywhere}), if $\Pol(\Y)\to\Dminion$ it must hold that $\CSP(\Y)$ is $\NP$-complete. It follows from the CSP Dichotomy Theorem proved in~\cite{Zhuk20:jacm,Bulatov17:focs} (in its reformulation in terms of polymorphism minions from~\cite{BOP18}) that the converse implication is also true.
\begin{thm}
[\cite{BOP18,Zhuk20:jacm,Bulatov17:focs}]
\label{thm_NP_completeness_projection_minion}
    Let $\Y$ be a  relational structure. Then $\CSP(\Y)$ is $\NP$-complete if $\Pol(\Y)\to\Dminion$, and it is tractable in polynomial time otherwise.
\end{thm}
As a consequence, we obtain the following sufficient condition for the occurrence of quantum advantage.
\begin{cor}
\label{cor_quantum_trivial_implies_tractable}
Let $\Y$ be a relational structure such that
$\CSP(\Y)$ is not tractable in polynomial time. If $\dim(\HH)\geq 3$, $\Y$ has quantum advantage over $\HH$.
\end{cor}

\begin{proof}
If $\CSP(\Y)$ is not in $\operatorname{P}$, it follows from Theorem~\ref{thm_NP_completeness_projection_minion} that $\Pol(\Y)\to\Dminion$. If $\dim(\HH)\geq 3$, we must then have $\Qminion_\HH\not\to\Pol(\Y)$, as otherwise composing the homomorphisms would yield $\Qminion_\HH\to\Dminion$, contradicting Proposition~\ref{prop_only_if_part_thm_dictator_quantum}. Hence, the conclusion follows from Theorem~\ref{thm_quantum_trivial_iff_minion_homo}.
\end{proof}
Equivalently, Corollary~\ref{cor_quantum_trivial_implies_tractable} may be phrased as follows: If $\CSP(\Y)$ is not in $\operatorname{P}$ and $\dim(\HH)\geq 3$,
there must exist some $\X$ such that $\X\qTo{\HH}\Y$ but $\X\not\to\Y$.

\section{$\Qminion_\HH$ and the bounded-width minion}
As we have seen, no structure has quantum advantage over Hilbert spaces of dimension $1$ or $2$. If the dimension is at least $3$, Corollary~\ref{cor_quantum_trivial_implies_tractable} provides a sufficient condition for the occurrence of quantum advantage, in terms of the complexity of $\CSP(\Y)$. In this section, we establish a necessary condition, also formulated in terms of a complexity notion of $\CSP(\Y)$ --- namely, its \emph{width}. 

The concept of width is central in the theory of constraint satisfaction. 
On a high level, the width of a CSP expresses the power of local-consistency techniques for its solvability. More formally, given two similar relational structures $\X$ and $\Y$ and some fixed $k\in\N$,
the $k$-consistency algorithm provides a heuristic to check if $\X\to\Y$, by testing for the existence of a non-empty family 
of partial homomorphisms from substructures of $\X$ of size at most $k$ to $\Y$, that is closed under restriction and under extension up to size $k$. If such a family exists, $\X$ and $\Y$ are said to be $k$-consistent.\footnote{We refer to~\cite{Barto14:jacm} for the formal definition of the $k$-consistency algorithm.} 
As long as $k$ is a constant, there exists a polynomial-time procedure (in the size of $\X$) to check whether $\X$ and $\Y$ are $k$-consistent. If $\X\to\Y$, a family of partial homomorphisms as described above always exists.
We say that $\CSP(\Y)$ has width $k$ if the $k$-consistency algorithm always detects unsatisfiable instances; i.e.,
if $\X$ and $\Y$ are $k$-consistent precisely when $\X\to\Y$. If that is the case, this method yields a polynomial-time algorithm for the solution of $\CSP(\Y)$. CSPs that are solvable in some constant width $k\in\N$ are known as \emph{bounded-width} CSPs, and constitute an important fragment of all tractable CSPs. 

The goal of this section is to prove that the width of a CSP parameterised by a structure having quantum advantage over some finite-dimensional Hilbert space must be unbounded (Corollary~\ref{cor_bounded_width_no_quantum_advantage}). 
Our proof is in some sense specular to the argument in the previous section. Indeed, this time we need to show that the quantum minion $\Qminion_\HH$ is located ``left enough'' in the minion homomorphism preorder, in that it always admits a homomorphism to one specific minion --- recently introduced in~\cite{cz23soda:minions} and, independently, in~\cite{bgs_robust23stoc} --- that is known to capture bounded width. Then, the result will follow from Theorem~\ref{thm_quantum_trivial_iff_minion_homo}. We shall conclude the section by outlining an alternative argument to obtain the same result, that makes use of a different minion introduced in~\cite{cz23sicomp:clap}. 

We start by defining the minion for bounded width.
\begin{defn}[\cite{cz23soda:minions,bgs_robust23stoc}]
$\Sminion$ is the linear minion over the abelian monoid $(\R,+)$ whose elements are all real matrices $M$ of finite but arbitrary size, such that $MM^{\top}$ is a diagonal matrix of trace $1$.      
\end{defn}
The fact that the linear minion $\Sminion$ captures bounded width follows from its connection to 
a CSP algorithm known as the \emph{standard semidefinite programming} relaxation (SDP). Essentially, SDP relaxes a given CSP by considering a semidefinite program whose variables are real vectors having constrained inner products. We now present the formal description of $\SDP$, following~\cite{Raghavendra08:everycsp}. 
Let $\X$ and $\Y$ be two $\sigma$-structures.
We consider a real vector $\bu_{x,y}$ for each $x\in X$, $y\in Y$, and a real number $v_{R,\bx,\by}$ for each $R\in\sigma$, $\bx\in R^\X$, and $\by\in R^\Y$. We also fix a real vector $\bu_0$ of norm $1$.  Consider the system
\begin{align*}
\begin{array}{lllll}
     \displaystyle\sum_{\by\in R^\Y}v_{R,\bx,\by}=1& (R\in\sigma,\bx\in R^\X)\\[5pt]
    \displaystyle v_{R,\bx,\by}\geq 0 & (R\in\sigma,\bx\in R^\X, \by\in R^\Y)\\[5pt]
\displaystyle \ang{\bu_{x_i,y},\bu_0}=\sum_{\substack{\by\in R^\Y\\y_i=y}}v_{R,\bx,\by}&
\left(
\begin{array}{llll}
     R\in\sigma,\bx\in R^\X,  \\
     i\in[\ar(R)],y\in Y 
\end{array}
\right)
\\[5pt]
\displaystyle \ang{\bu_{x_i,y},\bu_{x_j,y'}}=\sum_{\substack{\by\in R^\Y\\y_i=y,\,y_j=y'}}v_{R,\bx,\by}&
\left(
\begin{array}{llll}
     R\in\sigma,\bx\in R^\X,  \\
     i,j\in[\ar(R)],y,y'\in Y 
\end{array}
\right).
\end{array}
\end{align*}
We say that SDP \emph{solves} $\CSP(\Y)$ if $\X\to\Y$ whenever the system above admits a feasible solution such that all vectors $\bu_{x,y}$ belong to a real vector space of some finite, large-enough dimension.

\begin{thm}[\cite{cz23soda:minions,bgs_robust23stoc}]
\label{thm_SDP_minion}
Given a relational structure $\Y$, $\SDP$ solves $\CSP(\Y)$ if and only if $\Sminion\to\Pol(\Y)$.
\end{thm}
In fact, the theorem above corresponds to the fact that $\SDP$ has precisely the same power as $\Test{\Sminion}{}$ (i.e., the minion test associated with the minion $\Sminion$; cf.~\cite[\S~6.3.2]{bgs_robust23stoc}). The following result characterises the class of CSPs solved by SDP in terms of their width.
\begin{thm}[\cite{Barto16:sicomp}]
\label{thm_barto_kozik_SDP_bounded_width}
Given a relational structure $\Y$,
$\SDP$ solves $\CSP(\Y)$ if and only if $\CSP(\Y)$ has bounded width.
\end{thm}

As a consequence of the two results stated above, Corollary~\ref{cor_bounded_width_no_quantum_advantage} would follow if we show that the quantum minion admits a homomorphism to $\Sminion$.
To that end, it shall be useful to introduce a complex version of $\Sminion$.
\begin{defn}
$\Sminion_\C$ is the linear minion over the abelian monoid $(\C,+)$ whose elements are all complex matrices $M$ of finite but arbitrary size, such that $MM^*$ is a diagonal matrix of trace $1$.      
\end{defn}
It turns out that $\Sminion$ and $\Sminion_\C$ are in fact the same minion, up to homomorphic equivalence.
\begin{prop}
\label{prop_real_complex_SDP}
    $\Sminion$ and $\Sminion_\C$ are homomorphically equivalent.
\end{prop}
\begin{proof}
    The inclusion map yields a homomorphism $\Sminion\to\Sminion_\C$. To build a homomorphism in the opposite direction, take a matrix $M\in\Sminion_\C^{(\ell)}$ for some $\ell\in\N$, and
    let $M_R$ and $M_I$ denote the real and imaginary parts of $M$, respectively; i.e., $M_R$ and $M_I$ are real matrices of the same size as $M$, such that $M=M_R+ \mathrm{i}M_I$. We claim that the map $\vartheta:M\mapsto\begin{bmatrix}
        M_R&M_I
    \end{bmatrix}$ yields a minion homomorphism from $\Sminion_\C$ to $\Sminion$. Observe that
    \begin{align*}
        \vartheta(M)\vartheta(M)^\top
        =
        M_RM_R^{\top}+M_IM_I^{\top}.
    \end{align*}
    Furthermore,
    \begin{align*}
        MM^*
        &=
        (M_R+ \mathrm{i}M_I)(M_R+ \mathrm{i}M_I)^*
        =
        (M_R+ \mathrm{i}M_I)(M_R^{\top}- \mathrm{i}M_I^{\top})\\
        &=(M_RM_R^{\top}+M_IM_I^{\top})+\mathrm{i}(M_IM_R^{\top}-M_R{M_I}^\top).
    \end{align*}
    Since $MM^*$ is diagonal and has trace $1$, it follows that also its real part $M_RM_R^{\top}+M_IM_I^{\top}$ is diagonal and has trace $1$. Hence, we conclude that $\vartheta(M)\in\Sminion^{(\ell)}$. The fact that $\vartheta$ preserves the arity is clear. To show that it preserves the minors, we simply observe that, for $\pi:[\ell]\to[\ell']$, 
    \begin{align*}
        \vartheta(M_{/\pi})
        &=
        \begin{bmatrix}
           (M_{/\pi})_R&(M_{/\pi})_I 
        \end{bmatrix}
        =
        \begin{bmatrix}
           (M_R)_{/\pi}&(M_I)_{/\pi} 
        \end{bmatrix}\\
        &=
        \begin{bmatrix}
        M_R&M_I
        \end{bmatrix}_{/\pi}
        =
        \vartheta(M)_{/\pi},
    \end{align*}
    as required.
\end{proof}
We shall need the following basic properties of orthogonal projectors.
Recall that, for $p\in\End_\HH$,
$\rg_p$ denotes the range of $p$. 
\begin{lem}[\cite{hogben2013handbook}]
\label{lem_basic_projectors}
    Take $p,p'\in\Proj_\HH$.
    \begin{itemize}
        \item $pp'\in\Proj_\HH$ if and only if $[p,p']=\zeroOp_\HH$.
        \item $p+p'\in\Proj_\HH$ if and only if $pp'=p'p=\zeroOp_\HH$ if and only if $\rg_p\perp\rg_{p'}$. In this case, $\rg_{p+p'}=\rg_p\bplus\rg_{p'}$.
    \end{itemize}
\end{lem}

We now show that the quantum minion (over any $\HH$) homomorphically maps to $\Sminion$. Essentially, the proof works by fixing an arbitrary vector in $\HH$ and projecting it onto the spaces constituting the entries of the elements of $\Qminion_\HH$. This yields a homomorphism $\Qminion_\HH\to\Sminion_\C$ which, by Proposition~\ref{prop_real_complex_SDP}, is enough to establish the result. 

\begin{thm}
\label{thm_minion_homo_quantum_minion_skeletal_minion}
    $\Qminion_\HH\to\Sminion$.
\end{thm}
\begin{proof}
Fix an orthonormal basis $B$ of $\HH$, and let $d=\dim(\HH)$. Given any $\bv\in\HH$, we shall let $r(\bv)$ be the column vector of length $d$ containing the coordinates of $\bv$ in the basis $B$. Fix also an arbitrary vector $\bw\in\HH$ of norm $1$. For any $\bq\in\Qminion_\HH$ of some arity $\ell$, consider the complex $\ell\times d$ matrix $M$ defined by
\begin{align*}
    M^*\be_{i;\ell}
    =
    r(\pr_{q_i}(\bw))
\end{align*}
for each $i\in[\ell]$. We claim that the assignment $\xi:\bq\mapsto M$ yields a homomorphism from $\Qminion_\HH$ to $\Sminion_\C$.

First of all, we show that $M\in\Sminion_\C$. For $i,j\in[\ell]$, the $(i,j)$-th entry of $MM^*$ is
\begin{align*}
    \be_{i;\ell}^*MM^*\be_{j;\ell}
    =
    (r(\pr_{q_i}(\bw)))^*r(\pr_{q_j}(\bw))
    =
    \ang{\pr_{q_j}(\bw),\pr_{q_i}(\bw)}.
\end{align*}
Since $q_i\perp q_j$ whenever $i\neq j$, we deduce that $MM^*$ is a diagonal matrix. Therefore, we can write its trace as
\begin{align*}
    \tr(MM^*)
    &=
    \sum_{i,j\in[\ell]}\be_{i;\ell}^*MM^*\be_{j;\ell}
    =
    \sum_{i,j\in[\ell]}\ang{\pr_{q_j}(\bw),\pr_{q_i}(\bw)}\\
    &=
    \bigg{\|}\sum_{i\in[\ell]}\pr_{q_i}(\bw)\bigg{\|}^2.
\end{align*}
Furthermore, by Lemma~\ref{lem_basic_projectors}, 
\begin{align*}
    \sum_{i\in[\ell]}\pr_{q_i}(\bw)
    =
    \pr_{\bigboxplus_{i\in[\ell]}q_i}(\bw)
    =
    \pr_{\HH}(\bw)
    =\bw.
\end{align*}
It follows that $\tr(MM^*)=\|\bw\|^2=1$, as required to conclude that $M\in\Sminion_\C$.

We are left to show that $\xi$ is a minion homomorphism.
The fact that it preserves the arity directly follows from its definition.
To prove that it preserves the minors, take a map $\pi:[\ell]\to[\ell']$. We need to show that $\xi(\bq_{/\pi})=\xi(\bq)_{/\pi}$.
Pick an index $j\in [\ell']$, and observe that
\begin{align*}
    (\xi(\bq_{/\pi}))^*\be_{j;\ell'}
    &=
    r(\pr_{(\bq_{/\pi})_j}(\bw))
    =
    r(\pr_{\bigboxplus_{i\in\pi^{-1}(j)}q_i}(\bw))\\
    &=
    r\left(\sum_{i\in\pi^{-1}(j)}\pr_{q_i}(\bw)\right)
    =
    \sum_{i\in\pi^{-1}(j)}r(\pr_{q_i}(\bw))\\
    &=
     \sum_{i\in\pi^{-1}(j)}(\xi(\bq))^*\be_{i;\ell}
     =
     (\xi(\bq)_{/\pi})^*\be_{j;\ell'},
\end{align*}
thus yielding $\xi(\bq_{/\pi})=\xi(\bq)_{/\pi}$, as required.

It follows that $\xi$ yields a homomorphism from $\Qminion_\HH$ to $\Sminion_\C$, as claimed. To conclude that $\Qminion_\HH\to\Sminion$, we apply Proposition~\ref{prop_real_complex_SDP}.
\end{proof}

\begin{cor}
\label{cor_bounded_width_no_quantum_advantage}
Let $\Y$ be a relational structure such that
$\CSP(\Y)$ has bounded width.\footnote{It is known that a CSP is solvable in bounded width if and only if it is solvable in \emph{sublinear} width~\cite{AtseriasO19}.
Hence, in the statement of Corollary~\ref{cor_bounded_width_no_quantum_advantage}, the word ``bounded'' can be replaced by the word ``sublinear''.} Then $\Y$ does not have quantum advantage over any $\HH$.
\end{cor}

\begin{proof}
Suppose that $\CSP(\Y)$ has bounded width. Combining Theorems
~\ref{thm_SDP_minion} and~\ref{thm_barto_kozik_SDP_bounded_width}, we deduce that $\Sminion\to\Pol(\Y)$. Composing the latter homomorphism with the one from Theorem~\ref{thm_minion_homo_quantum_minion_skeletal_minion}, we find that $\Qminion_\HH\to\Pol(\Y)$ for any $\HH$. Then, the conclusion follows from Theorem~\ref{thm_quantum_trivial_iff_minion_homo}. 
\end{proof}

\begin{rem}
\label{rem_CLP_skeletal}
We now discuss an alternative proof of the result stated as Corollary~\ref{cor_bounded_width_no_quantum_advantage}, that makes use of a different minion and a different algorithm (both introduced in~\cite{cz23sicomp:clap}) instead of the minion $\Sminion$ and the algorithm SDP. 

The algorithm is the \emph{constraint linear programming} relaxation (CLP). Informally, when applied to a pair $\X$, $\Y$ of $\sigma$-structures,
CLP works by running multiple times the LP relaxation of the CSP, each time adding a different extra constraint that corresponds to fixing an assignment $\bx\mapsto\by$, where $\bx\in R^\X$, $\by\in R^\Y$, and $R\in\sigma$.
If the feasible region of the LP is empty, the assignment $\bx\mapsto\by$ is marked as unfeasible and removed from the space of solutions. The procedure continues until a fixed point is reached; then, the output is $\YES$ if and only if for each $\bx$ there is at least one feasible assignment.
CLP is at least as strong as the \emph{singleton arc consistency} relaxation introduced in~\cite{DB97} (see also~\cite{BD08,ChenDG13}); thus, in particular, it solves all bounded-width CSPs~\cite{Kozik21:sicomp}.  

We now describe the minion $\Cminion$ capturing CLP.
Let $M$ be a real matrix having $\ell$-many rows. We say that $M$ is \emph{skeletal} if, for any $i\in [\ell]$, whenever the $i$-th row of $M$ is not identically zero there exists a column of $M$ equal to $\be_{i;\ell}$.
We say that $M$ is \emph{stochastic} if it is entrywise nonnegative and each of its columns sums up to $1$. 
$\Cminion$ is defined as the linear minion over the abelian monoid $(\R,+)$ whose elements are all real skeletal stochastic matrices having finitely many rows.
It was shown in~\cite{cz23sicomp:clap} that $\Cminion$ captures the power of CLP, in the sense that CLP solves $\CSP(\Y)$ if and only if $\Cminion\to\Pol(\Y)$.\footnote{The minion $\Cminion$ described here differs from the CLP minion in~\cite{cz23sicomp:clap} (which we shall denote by $\tilde\Cminion$) in two technical aspects, that are both inessential. 
The first is that the matrices in $\tilde\Cminion$ are required to have rational (as opposed to real) elements. The reason why this difference is inessential is that running CLP on a CSP instance amounts to solving multiple LPs whose constraints are linear inequalities having rational coefficients. The columns of the matrices in the minion correspond to the solutions of the LPs. For such LPs,
a rational solution exists if and only if a real solution exists. Hence, considering rational- or real-valued matrices is equivalent in order to capture the power of CLP.
The second difference is that each matrix in $\tilde\Cminion$ is required to have countably many columns, that are eventually all equal. Nevertheless, using a compactness argument akin to the one in~\cite[\S~4.2]{cz23sicomp:clap}, it follows that, for any relational structure $\Y$, $\Cminion\to\Pol(\Y)$ if and only if $\tilde\Cminion\to\Pol(\Y)$.}

As a consequence, an alternative way of proving Corollary~\ref{cor_bounded_width_no_quantum_advantage} is by showing that $\Qminion_\HH\to\Cminion$ for each $\HH$. Let $S$ be the set of vectors in $\HH$ of norm $1$. Given $\bq\in\Qminion_\HH$ of some arity $\ell$, consider the matrix $M$ whose columns are indexed by the elements of $S$, defined as follows: For $\bh\in S$, the $\bh$-th column of $M$ is the vector
\begin{align*}
    \begin{bmatrix}
        \|\pr_{q_1}(\bh)\|^2&
        \|\pr_{q_2}(\bh)\|^2&
        \hdots&
        \|\pr_{q_\ell}(\bh)\|^2
    \end{bmatrix}^\top.
\end{align*}
It is straightforward to check that $M$ is stochastic.
The fact that $M$ is skeletal follows by observing that, if $\pr_{q_i}(\bh)\neq\bzero_\HH$ for some $i\in[\ell]$ and some $\bh\in S$, it must be that $q_i\neq\{\bzero_\HH\}$. Given a vector $\bh'\in q_i\cap S$, we find that the $\bh'$-th column of $M$ is $\be_{i;\ell}$, as needed. Therefore, $M\in\Cminion$. Moreover, the map $\bq\mapsto M$ is easily seen to preserve the arity and the minors, thus yielding $\Qminion_\HH\to\Cminion$.

\end{rem}

\section{Quantum advantage for graphs}
\label{sec_quantum_advantage_graphs}

In the case of graphs, the sufficient and necessary conditions for quantum advantage established in the previous two sections turn out to collapse. As a consequence, we obtain the following complete characterisation of quantum advantage for graphs, thus answering Question~\ref{question_quantum_advantage_graphs}.\footnote{We point out that the ``only if'' part of Theorem~\ref{thm_quantum_trivial_iff_bipartite} can also be derived from the results in~\cite{MancinskaRoberson16}, see Remark~\ref{rem_MR_homomorphisms}.}
\begin{thm}
\label{thm_quantum_trivial_iff_bipartite}
    If $\dim(\HH)\geq 3$, a graph has quantum advantage over $\HH$ if and only if it is  non-bipartite.
\end{thm}
(Note that no graph has quantum advantage over $\HH$ if $\dim(\HH)\leq 2$, see Corollary~\ref{cor_no_quantum_advantage_in_low_dimension}.)
We now show how the machinery developed in the previous sections proves Theorem~\ref{thm_quantum_trivial_iff_bipartite}.
The following, classic result by Hell and Ne\v{s}et\v{r}il classifies the complexity of graph CSPs.
\begin{thm}
[\cite{HellN90}]
\label{thm_hell_nesetril}
    Let $\Y$ be a graph. Then $\CSP(\Y)$ is in $\operatorname{P}$ if $\Y$ is bipartite, and it is $\NP$-complete otherwise.
\end{thm}
This is enough for proving Theorem~\ref{thm_quantum_trivial_iff_bipartite} under the assumption that $\operatorname{P}\neq\NP$, via Corollaries~\ref{cor_quantum_trivial_implies_tractable} and~\ref{cor_bounded_width_no_quantum_advantage}. To prove Theorem~\ref{thm_quantum_trivial_iff_bipartite} unconditionally, we make use of the following, revisited version of Hell--Ne\v{s}et\v{r}il's Theorem by Bulatov, that we formulate in terms of polymorphism minions.
\begin{thm}[\cite{Bulatov05}]
\label{thm_hell_nesetril_bulatov}
    Let $\Y$ be a graph. Then $\Pol(\Y)\to\Dminion$ if and only if $\Y$ is non-bipartite.
\end{thm}
\begin{proof}[Proof of Theorem~\ref{thm_quantum_trivial_iff_bipartite}]
Take a graph $\Y$, and suppose that $\Y$ is bipartite. If $\Y$ contains no edges, the result is trivial. Otherwise, $\Y$ is homomorphically equivalent to the graph $\K_2$ consisting of a single, undirected edge. It is well known that $\CSP(\K_2)$ has bounded width~\cite{Feder98:monotone}; hence, the same holds for $\CSP(\Y)$. By Corollary~\ref{cor_bounded_width_no_quantum_advantage}, we conclude that $\Y$ does not have quantum advantage over any $\HH$.

Suppose now that $\Y$ is non-bipartite. Since we are assuming that $\dim(\HH)\geq 3$, it follows from Theorems~\ref{thm_quantum_minion_vs_projections} and~\ref{thm_hell_nesetril_bulatov} that $\Qminion_\HH\not\to\Pol(\Y)$. Applying Theorem~\ref{thm_quantum_trivial_iff_minion_homo}, we deduce that $\Y$ has quantum advantage over $\HH$.
\end{proof}
\begin{rem}
\label{rem_MR_homomorphisms}
    The notion of quantum perfect strategy and quantum homomorphism from~\cite{abramsky2017quantum}  that we adopt in this paper --- which is applicable to arbitrary relational structures --- slightly differs from the one introduced in~\cite{MancinskaRoberson16} for the case of graphs. We now show that Theorem~\ref{thm_quantum_trivial_iff_bipartite} also holds if we make use of that alternative definition (which we shall indicate by the initials ``MR'').

    For two graphs $\X$ and $\Y$, an MR quantum homomorphism from $\X$ to $\Y$ over $\HH$ is a family of orthogonal projectors $\{p_{x,y}\}_{x\in X,y\in Y}$ in $\Proj_\HH$ satisfying the conditions (Q1) and (Q2) from Section~\ref{subsec_prelim_quantum_advantage}. In other words, commutativity between the projectors corresponding to adjacent vertices of $\X$ is not enforced in the MR definition (cf.~\cite[Corollary~2.2]{MancinskaRoberson16}).
    It easily follows from the definitions that $\X\qTo{\HH}\Y$ implies $\X\qToMR{\HH}\Y$. Therefore, if a graph $\Y$ has quantum advantage over $\HH$, it also has MR quantum advantage over $\HH$. Consequently, the fact that non-bipartite graphs have MR quantum advantage when $\dim(\HH)\geq 3$
    immediately follows from Theorem~\ref{thm_quantum_trivial_iff_bipartite}. The converse direction of the statement of the theorem in the MR setting (namely, that bipartite graphs have no MR quantum advantage) was proved in~\cite{MancinskaRoberson16}.
\end{rem}

\section{Quantum advantage and promise CSPs}
\label{sec_promise_CSPs}
Let $\Y$ be a relational structure and, as usual, let $\HH$ be a finite-dimensional Hilbert space. If $\Y$ does not have quantum advantage over $\HH$, any structure $\X$ such that $\X\qTo{\HH}\Y$  admits a homomorphism to $\Y$. Suppose now that $\Y$ does have quantum advantage over $\HH$. Can we still get some partial information of classical type from the fact that a quantum homomorphism $\X\qTo{\HH}\Y$ exists? 
A natural way of expressing this partial information --- as opposed to the complete information ``$\X\to\Y$'' --- is by asking that a quantum homomorphism $\X\qTo{\HH}\Y$  should at least guarantee the existence of a classical homomorphism $\X\to\Y'$ for \emph{some} structure $\Y'$. Notice that, if this is the case, we must have $\Y\to\Y'$, since $\Y\qTo{\HH}\Y$. We now give a formal description of this more general version of quantum advantage.

\begin{defn}
    Let $\Y$ and $\Y'$ be two $\sigma$-structures such that $\Y\to\Y'$. We say that the pair $(\Y,\Y')$ \emph{has quantum advantage over $\HH$} if there exists a $\sigma$-structure $\X$ such that $\X\qTo{\HH}\Y$ but $\X\not\to\Y'$.
\end{defn}
In other words, $(\Y,\Y')$ does \emph{not} have quantum advantage over $\HH$ if and only if a quantum homomorphism  $\X\qTo{\HH}\Y$ guarantees at least a classical homomorphism $\X\to\Y'$. 
Since the composition of a quantum and a classical homomorphism is a quantum homomorphism, we easily see that, if $(\Y,\Y')$ has quantum advantage over $\HH$, then the same must hold for both $\Y$ and $\Y'$ individually. We will prove that the converse implication is not true in general.
\begin{thm}
\label{thm_separation_quantum_advantage_one_two}
If $\dim(\HH)\geq 3$,
there exists a pair of relational structures $\Y\to\Y'$ such that $\Y$ and $\Y'$, but not $(\Y,\Y')$, have quantum advantage over $\HH$. 
\end{thm}

It turns out that the framework we developed in the previous sections can be smoothly readapted to capture this generalised version of quantum advantage. As a consequence, just like the occurrence of quantum advantage for a single structure is linked to CSP complexity, we shall see that the occurrence of quantum advantage for a pair of structures is linked to the complexity of a \emph{promise} version of CSPs that we now describe.

For two $\sigma$-structures $\Y$ and $\Y'$ such that $\Y\to\Y'$, the \emph{promise $\operatorname{CSP}$} parameterised by $\Y$ and $\Y'$ (in symbols, $\PCSP(\Y,\Y')$) is the following computational problem: Given as input a $\sigma$-structure $\X$, output $\YES$ if $\X\to\Y$, and $\NO$ if $\X\not\to\Y'$.
The fact that $\Y\to\Y'$ ensures that the answer sets are disjoint. Notice that $\PCSP(\Y,\Y)$ is precisely $\CSP(\Y)$, so PCSPs are a generalisation of CSPs. 

This framework was formally introduced in~\cite{AGH17} and~\cite{BG21} to study approximability of perfectly satisfiable CSPs, but particular examples of PCSPs have been studied for a long time. For instance, given two natural numbers $n\leq n'$, $\PCSP(\K_n,\K_{n'})$ is the famous \emph{approximate graph colouring} problem, whose complexity is not known for general $n$ and $n'$ --- while the complexity of its non-promise version, graph $n$-colouring, was already classified in~\cite{Karp72}. 

Unlike for CSPs, the complexity landscape of PCSPs is widely unknown. In particular, a dichotomy for PCSP complexity is not known to hold. Nevertheless, part of the machinery from the algebraic approach to CSPs does extend to the promise setting. Following~\cite{AGH17}, we let a polymorphism of arity $\ell$ for the pair $(\Y,\Y')$ be a homomorphism from $\Y^{\ell}$ to $\Y'$. The set $\Pol(\Y,\Y')$ of all polymorphisms of $(\Y,\Y')$ is a subset of $\mathscr{F}_{Y,Y'}$ (from Example~\ref{example_function_minion}) and it is closed under minors. Hence, just like $\Pol(\Y)$, $\Pol(\Y,\Y')$ is a (function) minion.

The next result generalises Theorem~\ref{thm_quantum_trivial_iff_minion_homo}, by showing that the quantum minion also captures quantum advantage for pairs of structures.
\begin{thm}
\label{thm_quantum_trivial_iff_minion_homo_PCSP}
Let $\Y\to\Y'$ be relational structures. Then
$(\Y,\Y')$ has quantum advantage over $\HH$ if and only if $\Qminion_\HH\not\to\Pol(\Y,\Y')$.
\end{thm}
Recall the notion of minion test used in Section~\ref{sec_quantum_minion}. Given a minion $\Mminion$, we say that $\Test{\Mminion}{}$ solves $\PCSP(\Y,\Y')$ if, for every structure $\X$ similar to $\Y$ and $\Y'$, $\Test{\Mminion}{}(\X,\Y)=\YES$ implies $\X\to\Y'$. Theorem~\ref{thm_minion_test_solvability_homo} extends to PCSPs (in fact, the following is its original formulation in~\cite{cz23soda:minions}).
\begin{thm}[\cite{cz23soda:minions}]
\label{thm_minion_test_solvability_homo_PCSPs}
    Let $\Mminion$ be a minion and let $\Y\to\Y'$ be relational structures. Then $\Test{\Mminion}{}$ solves $\PCSP(\Y,\Y')$ if and only if $\Mminion\to\Pol(\Y,\Y')$.
\end{thm}
The proof of Theorem~\ref{thm_quantum_trivial_iff_minion_homo_PCSP} is then entirely analogous to that of Theorem~\ref{thm_quantum_trivial_iff_minion_homo}.
\begin{proof}[Proof of Theorem~\ref{thm_quantum_trivial_iff_minion_homo_PCSP}]
The pair $(\Y,\Y')$ has quantum advantage over $\HH$ if and only if there exists some $\X$ such that $\Test{\Qminion_\HH}{}(\X,\Y)=\YES$ but $\X\not\to\Y'$; i.e., if and only if $\Test{\Qminion_\HH}{}$ does not solve $\PCSP(\Y,\Y')$. Hence, the conclusion follows from Theorem~\ref{thm_minion_test_solvability_homo_PCSPs}.
\end{proof}

The following statement is then a direct consequence of Theorems~\ref{thm_quantum_minion_vs_projections},~\ref{thm_minion_homo_quantum_minion_skeletal_minion}, and~\ref{thm_quantum_trivial_iff_minion_homo_PCSP}.
\begin{cor}
\label{cor_PCSP_quantum_advantage_summary}
    Let $\Y\to\Y'$ be relational structures. 
    \begin{enumerate}
        \item[$(1)$] If $\dim(\HH)\leq 2$, $(\Y,\Y')$ has no quantum advantage over $\HH$.
        \item[$(2)$] If $\dim(\HH)\geq 3$ and $\Pol(\Y,\Y')\to\Dminion$, $(\Y,\Y')$ has quantum advantage over $\HH$.
        \item[$(3)$] If $\Sminion\to\Pol(\Y,\Y')$, $(\Y,\Y')$ has no quantum advantage over any $\HH$.
    \end{enumerate}
\end{cor}

How to link the result above to PCSP complexity?
It was shown in~\cite{BBKO21} that Theorem~\ref{thm_minion_homos_reductions} keeps holding in the PCSP setting. I.e., a minion homomorphism $\Pol(\Y,\Y')\to\Pol(\tilde\Y,\tilde\Y')$ induces a polynomial-time (log-space) reduction from $\PCSP(\tilde\Y,\tilde\Y')$ to $\PCSP(\Y,\Y')$. Nevertheless, some care needs to be taken when reformulating the conditions in parts ($2$) and ($3$) of Corollary~\ref{cor_PCSP_quantum_advantage_summary} in terms of PCSP-complexity statements. Regarding part ($2$), in sharp contrast to the CSP setting,
it is not true that the polymorphism minion of
any $\NP$-hard PCSP admits a homomorphism to $\Dminion$, as detailed in~\cite{BBKO21}. Therefore, unless $\operatorname{P}=\NP$, we cannot extend Corollary~\ref{cor_quantum_trivial_implies_tractable} to the PCSP setting. 
\begin{example}
Given a graph $\X$, let its quantum chromatic number be the quantity $\chi_{\operatorname{q}}(\X)=\min\{n\in\N:\X\qTo{\HH}\K_n\}$. Then, for $n\leq n'\in\N$, quantum advantage for the pair $(\K_n,\K_{n'})$ corresponds to a separation between quantum and classical chromatic numbers. Indeed, $(\K_n,\K_{n'})$ has quantum advantage over $\HH$ precisely when there exists a graph having quantum chromatic number $\leq n$ and classical chromatic number $>n'$. It was shown in~\cite{BrakensiekG16} (see also~\cite{BBKO21}) that, given integers $3\leq n\leq n'$, $\Pol(\K_n,\K_{n'})\to\Dminion$ if and only if $n'\leq 2n-2$. Thus, if $\dim(\HH)\geq 3$ and $n\geq 3$, Corollary~\ref{cor_PCSP_quantum_advantage_summary} implies that there exist graphs having quantum chromatic number at most $n$ and classical chromatic number at least $2n-1$. However, no larger separation can be derived through Corollary~\ref{cor_PCSP_quantum_advantage_summary}.\footnote{We remark that this version of quantum chromatic number differs from the one introduced in~\cite{CameronMNSW07} (which we denote by $\tilde\chi_{\operatorname{q}}$), in that the latter makes use of projectors that are not necessarily commuting (cf.~Remark~\ref{rem_MR_homomorphisms}). In particular, for any graph $\X$, it holds that 
$\tilde\chi_{\operatorname{q}}(\X)\leq\chi_{\operatorname{q}}(\X)\leq\chi(\X)$, where $\chi$ is the classical chromatic number. 
An exponential separation between $\tilde\chi_{\operatorname{q}}$ and $\chi$ is known,  see~\cite{BuhrmanCW98,brassard1999cost,AvisHKS06}.}
\end{example}

As for part ($3$) of Corollary~\ref{cor_PCSP_quantum_advantage_summary}, the condition $\Sminion\to\Pol(\Y,\Y')$ does capture PCSP solvability via SDP; also in this case, the following, more general formulation of Theorem~\ref{thm_SDP_minion} is the one given in~\cite{cz23soda:minions,bgs_robust23stoc}.
\begin{thm}[\cite{cz23soda:minions,bgs_robust23stoc}]
\label{thm_SDP_minion_PCSP}
Let $\Y\to\Y'$ be relational structures.
Then $\SDP$ solves $\PCSP(\Y,\Y')$ if and only if $\Sminion\to\Pol(\Y,\Y')$. 
\end{thm}
As we see next, the result above allows finding a separation for the notions of quantum advantage for single structures and for pairs of structures.
\begin{proof}[Proof of Theorem~\ref{thm_separation_quantum_advantage_one_two}]
    Let $\textbf{Z}$ and $\textbf{Z}'$ be the relational structures encoding the CSPs ``\textsc{1-in-3 Sat}'' and ``\textsc{Not-All-Equal Sat}'', respectively. I.e., $\textbf{Z}$ and $\textbf{Z}'$ have Boolean domain and a unique, ternary relation 
    \begin{align*}
    R^\textbf{Z}&=\{(1,0,0),(0,1,0),(0,0,1)\},\\
    R^{\textbf{Z}'}&=\left\{\begin{array}{lll}
         (1,0,0),(0,1,0),(0,0,1),\\
         (0,1,1),(1,0,1),(1,1,0) 
    \end{array}
    \right\}.
    \end{align*}
It is well known~\cite{Garey79} that both $\CSP(\textbf{Z})$ and $\CSP(\textbf{Z}')$ are NP-complete problems, and it is straightforward to check that $\Pol(\textbf{Z})\to\Dminion$ and $\Pol(\textbf{Z}')\to\Dminion$. Hence, since we are assuming $\dim(\HH)\geq 3$, both $\textbf{Z}$ and $\textbf{Z}'$ have quantum advantage over $\HH$. On the other hand, it was shown in~\cite{bgs_robust23stoc} that $\PCSP(\textbf{Z},\textbf{Z}')$ is solved by $\SDP$. It follows from Theorem~\ref{thm_SDP_minion_PCSP} that $\Sminion\to\Pol(\textbf{Z},\textbf{Z}')$. Hence, Corollary~\ref{cor_PCSP_quantum_advantage_summary} implies that the pair $(\textbf{Z},\textbf{Z}')$ has no quantum advantage over $\HH$.
\end{proof}
\begin{rem}
\label{rem_1in3_NAE_CLP}
The fact that the pair $(\textbf{Z},\textbf{Z}')$ from the proof of Theorem~\ref{thm_separation_quantum_advantage_one_two} has no quantum advantage over any $\HH$ can also be obtained by using the minion $\Cminion$ from Remark~\ref{rem_CLP_skeletal} instead of $\Sminion$.
Indeed, also in this case, $\Cminion$ captures solvability by the CLP relaxation in the more general setting of PCSPs: It was proved in~\cite{cz23sicomp:clap} that CLP solves $\PCSP(\Y,\Y')$ if and only if $\Cminion\to\Pol(\Y,\Y')$. Since, as shown in Remark~\ref{rem_CLP_skeletal}, $\Qminion_\HH\to\Cminion$ for each $\HH$, part $(3)$ of Corollary~\ref{cor_PCSP_quantum_advantage_summary} also holds if we replace $\Sminion$ by $\Cminion$. It was shown in~\cite{BG21} that $\PCSP(\textbf{Z},\textbf{Z}')$ is solved by CLP; thus, $\Cminion\to\Pol(\textbf{Z},\textbf{Z}')$, which means that $(\textbf{Z},\textbf{Z}')$ has no quantum advantage over any $\HH$.

We also point out that, while a notion of width for PCSPs is known~\cite{BBKO21,Atserias22:soda},
neither of the minions $\Sminion$ and $\Cminion$ corresponds to bounded width for PCSPs. Indeed, as we have seen, $\Pol(\textbf{Z},\textbf{Z}')$ admits a homomorphism from both $\Sminion$ and $\Cminion$ but, as shown in~\cite{Atserias22:soda}, $\PCSP(\textbf{Z},\textbf{Z}')$ has unbounded width.
\end{rem}

\section*{Acknowledgements}
The author is grateful to David E. Roberson, Laura Man\v{c}inska, and Samson Abramsky for helpful conversations on quantum homomorphisms and quantum advantage.

{\small
\bibliographystyle{plainurl}
\bibliography{bibliography}
}

\appendix
\section{Omitted proofs}
\label{appendix_omitted_proofs}
\begin{prop*}
[Proposition~\ref{prop_quantum_homo_as_test} restated]
    Let $\X,\Y$ be similar relational structures. Then $\X\qTo{\HH}\Y$ if and only if $\X\to\freeQH(\Y)$.
\end{prop*}
The proof of Proposition~\ref{prop_quantum_homo_as_test} makes use of the following result.
\begin{lem}[\cite{prugovecki1982quantum}]
\label{lem_small_family_of_projectors_vanishes}
Let $\{p_1,p_2,\dots,p_\ell\}\subseteq\Proj_\HH$ be a finite family of orthogonal projectors. Then $\id_\HH-\sum_{i\in [\ell]}p_i\succeq \zeroOp_\HH$  if and only if $p_ip_j=\zeroOp_\HH$ for each $i\neq j\in [\ell]$.
\end{lem}

\begin{proof}[Proof of Proposition~\ref{prop_quantum_homo_as_test}]
Suppose that $\X\qTo{\HH}\Y$, and take a family $\{p_{x,y}\}_{x\in X,\,y\in Y}$ of orthogonal projectors in $\Proj_\HH$ witnessing it. Let $n$ be the domain size of $\Y$.
Consider the function $f$ mapping $x\in X$ to the vector of length $n$ whose $y$-th element is $\rg_{p_{x,y}}$ for each $y\in Y$.
We claim that $f(x)\in\Qminion_\HH^{(n)}$. By ($\operatorname{Q_1}$), the family $\{p_{x,y}\}_{y\in Y}$ satisfies the hypothesis of Lemma~\ref{lem_small_family_of_projectors_vanishes}; hence, $p_{x,y}p_{x,y'}=\zeroOp_\HH$ for each $y\neq y'\in Y$. Then, a repeated application of the second part of Lemma~\ref{lem_basic_projectors} yields that $\rg_{p_{x,y}}\perp\rg_{p_{x,y'}}$ for each $y\neq y'\in Y$, and 
\begin{align*}
    \bigboxplus_{y\in Y}\rg_{p_{x,y}}
    =
    \rg_{\sum_{y\in Y}p_{x,y}}
    =
    \rg_{\id_\HH}
    =
    \HH.
\end{align*}
This proves the claim, and shows that $f:X\to\Qminion_\HH^{(n)}$ is a well-defined function. We now claim that $f$ yields a homomorphism from $\X$ to $\freeQH(\Y)$. To that end, take a symbol $R$ of arity $r$ in the common signature $\sigma$ of $\X$ and $\Y$, and consider a tuple $\bx\in R^\X$. We need to prove that $f(\bx)\in R^{\freeQH(\Y)}$ (where $f(\bx)$ denotes the entrywise application of $f$ to the entries of $\bx$). Consider the vector $\bw$ of length $m=|R^\Y|$ whose $\by$-th entry, for $\by\in R^\Y$, is $\rg_{p_{\bx,\by}}$ (where we recall that $p_{\bx,\by}=p_{x_1,y_1}\cdot p_{x_2,y_2}\cdot\ldots\cdot p_{x_r,y_r}$).
The condition ($\operatorname{Q_2}$) guarantees that the factors in $p_{\bx,\by}$ are pairwise commuting; hence,  by the first part of Lemma~\ref{lem_basic_projectors},
we find $p_{\bx,\by}\in\Proj_\HH$. Notice also that 
\begin{align*}
    \sum_{\by\in R^\Y}p_{\bx,\by}
    &=
    \sum_{\by\in Y^r}p_{\bx,\by}
    =
    \left(\sum_{y\in Y}p_{x_1,y}\right)
    \left(\sum_{y\in Y}p_{x_2,y}\right)
    \dots
    \left(\sum_{y\in Y}p_{x_r,y}\right)
    =\id_\HH^r
    =\id_\HH,
\end{align*}
where the first equality comes from ($\operatorname{Q_3}$) and the third from ($\operatorname{Q_1}$). We deduce through Lemma~\ref{lem_small_family_of_projectors_vanishes} that $p_{\bx,\by}p_{\bx,\by'}=\zeroOp_\HH$ whenever $\by\neq\by'$. It then follows from Lemma~\ref{lem_basic_projectors} that the entries of $\bw$ are pairwise orthogonal spaces summing up to $\HH$, which yields $\bw\in\Qminion_\HH^{(m)}$. 
Recall from Section~\ref{sec_quantum_minion} that, given $i\in [r]$, $\pi_i$ is the map assigning to each $\by\in R^\Y$ its $i$-th entry $y_i$.
We are left to show that $f(x_i)=\bw_{/\pi_i}$ for each $i\in [r]$. Given $y\in Y$, we find
\begin{align*}
    \sum_{\by\in\pi_i^{-1}(y)}p_{\bx,\by}
    &=
    \sum_{\substack{\by\in R^\Y\\y_i=y}}p_{\bx,\by}
    =
    \sum_{\substack{\by\in Y^r\\y_i=y}}p_{\bx,\by}
    =\id_\HH^{r-1}\cdot p_{x_i,y}
    =
    p_{x_i,y}.
\end{align*}
It follows that the $y$-th entry of $\bw_{/\pi_i}$ is
\begin{align*}
\bigboxplus_{\by\in\pi_i^{-1}(y)}\rg_{p_{\bx,\by}}
=
    \rg_{\sum_{\by\in\pi_i^{-1}(y)}p_{\bx,\by}}
    =
    \rg_{p_{x_i,y}},
\end{align*}
which equals the $y$-th entry of $f(x_i)$. This proves the claim and shows that $f(\bx)\in R^{\freeQH(\Y)}$, thus establishing that $f$ is indeed a homomorphism from $\X$ to $\freeQH(\Y)$.

To prove the converse implication, let $f:\X\to\freeQH(\Y)$ be a homomorphism. Recall that, for a linear subspace $v\in L_\HH$, $\pr_v\in\Proj_\HH$ denotes the orthogonal projector onto $v$; for typographical convenience, we shall denote $\pr_v$ by $\pr(v)$ in the rest of this proof. For $x\in X$ and $y\in Y$, let $p_{x,y}=\pr(f(x)_y)$, where $f(x)_y\in L_\HH$ is the $y$-th entry of the vector $f(x)$. 
We claim that the family $\{p_{x,y}\}_{x\in X, y\in Y}$ is a witness for the existence of a quantum homomorphism from $\X$ to $\Y$ over $\HH$. 

The condition ($\operatorname{Q_1}$) is easily satisfied as $\bigboxplus_{y\in Y}f(x)_y=\HH$. Take now $R\in\sigma$ of arity $r$, $\bx\in R^\X$, $k,\ell\in[r]$, and $y,y'\in Y$. Since $f$ is a homomorphism, there exists some vector $\bw\in\Qminion_\HH^{(m)}$ (where $m=|R^\Y|$) such that $f(x_i)=\bw_{/\pi_i}$ for each $i\in [r]$. Observe that
\begin{align}
\label{eqn_2024_01_12_2120}
    f(x_i)_y
    =
    (\bw_{/\pi_i})_y
    =
    \bigboxplus_{\by\in\pi_i^{-1}(y)}w_\by
    =
    \bigboxplus_{\substack{\by\in R^\Y\\y_i=y}}w_\by.
\end{align}
Hence, we have
\begin{align*}
    p_{x_k,y}
    &=
    \pr\bigg{(}\bigboxplus_{\substack{\by\in R^\Y\\y_k=y}}w_\by\bigg{)}
    =
    \sum_{\substack{\by\in R^\Y\\y_k=y}}\pr(w_\by)
    \intertext{and, similarly,}
    p_{x_\ell,y'}
    &=\sum_{\substack{\by\in R^\Y\\y_\ell=y'}}\pr(w_\by).
\end{align*}
Using the linearity of the commutator and the fact that the product of orthogonal projectors onto orthogonal spaces is zero (the second part of Lemma~\ref{lem_basic_projectors}), we deduce that 
\begin{align*}
[p_{x_k,y},p_{x_\ell,y'}]=\zeroOp_\HH,    
\end{align*}
which shows that the condition ($\operatorname{Q_2}$) holds. Take now $\bz\in Y^r\setminus R^\Y$. By~\eqref{eqn_2024_01_12_2120}, we have 
\begin{align*}
    p_{x_i,z_i}
    =
    \pr\bigg{(}\bigboxplus_{\substack{\by\in R^\Y\\y_i=z_i}}w_\by\bigg{)}
    =
    \sum_{\substack{\by\in R^\Y\\y_i=z_i}}\pr(w_\by)
\end{align*}
for each $i\in [r]$. 
Using again that the product of orthogonal projectors onto orthogonal spaces is zero and the fact that $\bz\not\in R^\Y$, we conclude that
\begin{align*}
    p_{\bx,\bz}
    =
    \prod_{i\in[r]}
     p_{x_i,z_i}
     =
     \prod_{i\in[r]}
     \sum_{\substack{\by\in R^\Y\\y_i=z_i}}\pr(w_\by)
     =
     \sum_{\substack{\by\in R^\Y\\\by=\bz}}\pr(w_\by)
     =
     \zeroOp_\HH,
\end{align*}
as needed to show that also the condition ($\operatorname{Q_3}$) is satisfied.
\end{proof}

\begin{lem*}
[Lemma~\ref{lem_Hopft_fibration} restated]
    If $\dim(\HH)=2$, there exists $C\subset \HH$ such that, given any two nonzero orthogonal
    vectors $\bv,\bw\in\HH$, exactly one of $\bv$ and $\bw$ lies in $C$.
\end{lem*}
\begin{proof}
    Consider the map 
    \begin{align*}
        f:\C^2&\to\R\times\C\\
        (x,y)&\mapsto(x\bar{x}-y\bar{y},2x\bar{y}).
    \end{align*}
    Let $S^3$ be the subset of $\C^2$ containing all pairs $(x,y)\in\C^2$ such that $x\bar{x}+y\bar{y}=1$, and let $S^2$ be the subset of $\R\times\C$ containing all pairs $(t,z)\in\R\times\C$ such that $t^2+z\bar{z}=1$.
    A straightforward calculation shows that $f$ maps elements in  $S^3$ to elements in $S^2$. Moreover, given two orthogonal vectors $(x,y)$ and $(p,q)$ in $S^3$ (i.e., $x\bar{p}+y\bar{q}=0$), notice that the matrix
    \begin{align*}
        M=\begin{bmatrix}
            x&p\\y&q
        \end{bmatrix}
    \end{align*}
    satisfies $M^*M=I_2$; i.e., $M$ is unitary. As a consequence,
    \begin{align*}
        I_2
        =
        MM^*
        =
        \begin{bmatrix}
            x\bar{x}+p\bar{p} & x\bar{y}+p\bar{q}\\
            y\bar{x}+q\bar{p} & y\bar{y}+q\bar{q}
        \end{bmatrix},
    \end{align*}
and we deduce that $f(x,y)=-f(p,q)$. In other words, $f$ maps pairs of orthogonal vectors of $\C^2$ having norm $1$ to antipodal vectors in $S^2$. 

Choose an antipodal partition of $S^2$; i.e., a partition $S^2=U\cup V$ such that, for each $(t,z)\in S^2$, $(t,z)$ and $(-t,-z)$ lie in different parts. For example, we might let $U$ be the set of vectors $(t,z)\in S^2$ such that $t>0$, or $t=0$ and $z_R>0$ (where $z_R$ is the real part of $z$), or $t=0$ and $z=\mathrm{i}$, and $V=S^2\setminus U$.
Fix an orthonormal basis $B$ for $\HH$, and let $r:\HH\to\C^2$ be the function mapping elements of $\HH$ to the corresponding vectors  of coordinates in the basis $B$. Finally, consider the function $g:\HH\setminus\{\bzero_\HH\}\to S^2$ mapping $\bv$ to the image of $\frac{\bv}{\|\bv\|}$ under the function $f\circ r$.  
We define $C=g^{-1}(U)$. If $\bv,\bw$ are two nonzero orthogonal vectors in $\HH$, the argument above implies that $g(\bv)=-g(\bw)$ and, hence, exactly one of $g(\bv)$ and $g(\bw)$ lies in $U$. This shows that the set $C$ has the required property.
\end{proof}

\end{document}